\documentclass[a4paper]{article}
\usepackage{graphicx} 
\usepackage{epstopdf}
\usepackage[caption=false]{subfig}

\usepackage[]{natbib}

\usepackage{layaureo}
\usepackage{changes}
\usepackage{textcomp}
\usepackage{amssymb}  
\usepackage{amsmath}  
\usepackage{amsthm}
\usepackage{mathabx}
\usepackage{color}    

\usepackage{tikz}
\usepackage{pgf}
\usepackage{bussproofs}
\usetikzlibrary{automata, positioning, arrows}
\tikzset{ ->, node distance=2cm, every state/.style={thick, fill=gray!10}, initial text=$ $, }
\usepackage[T1]{fontenc}
\usepackage{url}
\usepackage[latin1]{inputenc}
\usepackage{listings}
\usepackage{pgfplots}
\usepackage{proof}
\pgfplotsset{compat=1.16}
\usepackage{comment}

\newcommand\fst{\mathit{fst}}
\newcommand\snd{\mathit{snd}}
\newcommand\Heads{\mathit{Heads}}
\newcommand\Tails{\mathit{Tails}}

\newcommand\Trust{\mathit{Trust}}
\newcommand\UTrust{\mathit{UTrust}}
\newcommand\distribution{\mathit{distribution}}
\newcommand\Output{\mathit{output}}
\newcommand{\indep}{\perp \!\!\! \perp}

\newcommand{\listerm}{\mathcal{L}}

\newenvironment{scprooftree}[1]{\gdef\scalefactor{#1}\begin{center}\proofSkipAmount \leavevmode}%
  {\scalebox{\scalefactor}{\DisplayProof}\proofSkipAmount \end{center} }

\theoremstyle{plain}
\newtheorem{theorem}{Theorem}[section]

\newtheorem{corollary}[theorem]{Corollary}

\theoremstyle{definition}
\newtheorem{definition}[theorem]{Definition}
\newtheorem{example}[theorem]{Example}

\theoremstyle{remark}

\date{}

\begin{document}


\title{Checking trustworthiness of probabilistic computations in a typed natural deduction system}

\author{
F.~A. D'Asaro\textsuperscript{a,b,*}, F.~A. Genco\textsuperscript{a,c} and G. Primiero\textsuperscript{a}\\
\textsuperscript{a}LUCI Lab, Department of Philosophy, University of Milan, Italy;\\\textsuperscript{b}Department of Human Sciences, University of Verona, Verona, Italy
\\\textsuperscript{c}Center for Logic, Language and Cognition, University of Turin, Turin, Italy
\\\footnote{Corresponding author, \texttt{fabioaurelio.dasaro@univr.it}}
}

\maketitle

\begin{abstract}
In this paper we present the probabilistic typed natural deduction calculus TPTND, designed to reason about and derive trustworthiness properties of probabilistic computational processes, like those underlying current AI applications. Derivability in TPTND is interpreted as the process of extracting $n$ samples of possibly complex outputs with a certain frequency from a given categorical distribution. We formalize trust for such outputs as a form of hypothesis testing on the distance between such frequency and the intended probability. The main advantage of the calculus is to render such notion of trustworthiness checkable. We present a computational semantics for the terms over which we reason and then the semantics of TPTND, where logical operators as well as a Trust operator are defined through introduction and elimination rules. We illustrate structural and metatheoretical properties, with particular focus on the ability to establish under which term evaluations and logical rules applications the notion of trustworthiness can be preserved.
\end{abstract}


\section{Introduction}

The extensive use of AI techniques in decision-making systems is significantly increasing the need for verification of the safety and reliability properties of probabilistic computations. The possibility of formal verification of such properties would grant the ability to consider such systems trustworthy. Nowadays, several approaches to verification of AI systems are emerging, see \citep{3448248,urban2021review} for overviews. Notably, approaches are focusing on model-checking techniques for safety, liveness and fairness properties, see, e.g., \citep{3133904,97830300109042}, program analysis and synthesis, see, e.g., \citep{dreossi2019verifai} or proof-checkers see, e.g., \citep{1210044,97833199614533} for their increasing use  also at the industrial level. In the technical literature, though, little is available on the formalization of a notion of trustworthiness itself in this context. Current logical approaches to computational trust in general are overviewed in Section \ref{sec:related}.

In this paper, we introduce a derivation system dubbed \textit{Trustworthy Probabilistic Typed Natural Deduction} (TPTND for short). The aim of this system is to formalize the task of inferential reasoning \textit{about} probabilistic computational processes, in particular about their trustworthiness. We consider samples of such processes with corresponding frequencies and reason about their distance from corresponding theoretical probabilities. We start by defining an operational semantics to consider a probability space in which such computations can be thought to be evaluated. This operational semantics defines events whose results have a certain probability of occurring, and generalizes to samples of such events with frequencies of observed outputs. We then transform such terms under operations for update, conjunction, disjunction and dependency of outputs. A full probabilistic $\lambda$-calculus for formalising experiments consisting of several executions of a program and for evaluating trust is introduced in \citep{230200958}. This $\lambda$-calculus does not feature, however, any other primitive operation on probabilistic events. 

In the language of TPTND, reasoning about such probabilistic computations is given in the form of logical rules for deriving theoretical probabilities, expected probabilities and frequencies. 
Accordingly, judgements of our language are sequents of the form $\Gamma \vdash \phi$ where $\Gamma$ is a {\it context}, that is,  a set of \textit{assumptions} on random variables, and $\phi$ is a typed formula in one of the following forms:
\begin{itemize}
\item $x:\alpha_{a}$, declaring that the random variable $x$ has been assigned value $\alpha$ with theoretical probability $a$;
\item $t_{n}:\alpha_{\tilde{a}}$, declaring that $a$ is the expected probability of obtaining the value $\alpha$ in a sample of $n$ trials on the process $t$;
\item $t_{n}:\alpha_{f}$, declaring that $f$ is the frequency with which the value $\alpha$ has been obtained during a series of $n$ trials on the process $t$.
\end{itemize}

Both an element of a context and a derived expression may express a statement about a random variable. The derivability relation $\vdash$ expresses dependency between probabilities. For example, the judgement\[ x : \alpha \vdash y : \beta_{0.1} \]means that the random variable $y$ has probability $0.1$ of producing output $\beta$, provided $x$ has output $\alpha$. Under a context interpreted as representing a probability distribution, a derived formula expresses a sample extracted from such distribution. The type of the derived formula is then decorated with an expected probability and the term with a sample size. For example, the judgement\[\Gamma \vdash toss_{100} : \Heads_{\widetilde{0.35}} \]may be interpreted as the statement saying that, under distribution $\Gamma$, after tossing a coin $100$ times, one expects output $\Heads$ $35\%$ of the times. $\Gamma$ specifies the value assigned to the output of each random variable of interest in the context of an experiment consisting of coin tossing trials\footnote{We talk of an {\it experiment} simply in order to refer to a series of trials on a particular process. A {\it trial} on a process, in turn, is just an execution of the process.} (the possible outputs are thus $Heads$ and $Tails$). Finally, we can also express any concrete experiment consisting of a series of trials by displaying the frequency of a particular output during the experiment. For instance,\[\Gamma \vdash toss_{100} : {\Heads_{30/100}} \]which says that over $100$ tosses of our coin, the output $Heads$ was obtained $30$ times. 


The calculus defines also rules for Bayesian update and trust evaluation, where judgements occurring in the rule may include both probabilities and frequencies. In our example above, the theoretical probability of the coin assumed to be fair may be updated after a trial of $n$ experiments, empirically run by throwing the coin. And under the assumption that the coin is fair (or biased), one may want to judge the trustworthiness of the process at hand, e.g., we may ask whether the $toss$ process follows a fair distribution over the outcomes $Heads$ and $Tails$. For example, assuming that $\Gamma$ expresses the distribution corresponding to a fair coin $c$ as the possible output declarations of one random variable $\Gamma=\{x_{c}:Heads_{0.5}, x_{c}:Tails_{0.5}\}$, the process $toss$ should be considered trustworthy if, by increasing the number $n$ of trials on it, the probability of output $Heads$ gets arbitrarily closer to $0.5$ in the limit, and in particular when the displayed frequency of $Heads$ falls within a given confidence interval determined by an acceptable error rate. Note that the judgement above can also be formulated when the distribution denoted by $\Gamma$ is unknown, i.e., when one does not know whether the tossed coin is fair or biased, and the frequency value of the output $Heads$ is just empirically observed. Generalizing from this toy example, any probabilistic computation can be checked for trustworthiness in this sense. Therefore, our system can be employed to verify how much we should trust, for instance, a classifier to produce the correct output.

Meta-theoretical properties are formulated considering the computational transformations of empirical tests and their frequencies according to logical operations. In particular, we show which syntactic transformations of terms as they are denoted by inferences of the system
\begin{enumerate}
\item do not permit to infer unintended outputs for the processes; and
\item allow to infer whether the observed computation is trustworthy, and under which inferential steps such trustworthiness can be preserved.
\end{enumerate}

We consider this system a useful step towards the formal design and possible automation of reasoning about probabilistic computations like those in use in Machine Learning techniques, 
as well as for the development of proof-checking protocols on such computational systems. In the following, we show the internal workings of TPTND by the use of simple examples, like dice rolling and coin tossing. However, we also hint at concrete applications using a more realistic example inspired by gender-bias verification.

The rest of this paper is structured as follows. In Section \ref{sec:related} we overview the literature on proof-theoretic formal verification of probabilistic systems and computational trust.
In Section \ref{sec:compsec} we introduce a computational semantics defining how atomic processes of interest behave from an operational perspective 
and their closure under logical operations. In Section \ref{sec:system} we introduce our system TPTND, first through its syntax (Section \ref{subsec:syntax}), then through its inferential rules to reason: first about theoretical probabilities (Section \ref{subsec:reason_random}); then about processes or computations executed under given probability distributions (Section \ref{subsec:reason_processes}). We further show how to model Bayesian Updating of random variables in TPTND (Section \ref{subsec:bayesian}) and finally we introduce the inferential engine to reason about trust on the computational processes of interest (Section \ref{subsec:trust}) and to derive structural properties on the relevant probability distributions (Section \ref{subsec:structural}). In Section \ref{sec:meta} we prove the main meta-theoretical results about TPTND, aiming at showing that given a semantics of well-behaving terms in a probabilistic space, the derivability relation of our system guarantees that outputs obtained with sufficient level of trustworthiness are preserved, and accordingly side-effects are avoided. We conclude in Section \ref{sec:conclusions} indicating further steps of this research.

\section{Related Work}\label{sec:related}

The present work offers an example of a formal verification method for trustworthy probabilistic computations as implemented for example by machine-learning systems. This area of research, focusing on the modeling of probabilistic systems and reasoning over satisfiability of properties of relevance is yet in its infancy, see \citep{3448248}. Moreover, it mainly refers to the development of computational, temporal and hybrid logics as formal methods to prove trustworthy properties of interest, see, e.g., \citep{489079,647810738106,DBLP:conf/lori/TerminePD21}; or with linear temporal logic properties defined over Markov decision processes, e.g., with reinforcement learning methods \citep{3302509.3311053} or with imprecise probabilities \citep{DBLP:conf/eumas/TermineAPF21}. 

In this context, the logic introduced in this paper offers a novel combination of methods and approaches. First of all, it refers to the design of inferential systems of the family of natural deduction systems, sequent calculi and typed $\lambda$-calculi, extended with probabilities. These systems and the associated formal proof-verification methods have so far been only little explored in their ability to check trustworthiness of probabilistic computational systems. Secondly, we consider trustworthiness as a property which can itself be expressed and formalized explicitly in the calculus, and therefore openly checked for satisfiability or derivability. Hence, our language TPTND combines these two areas by presenting a probabilistic derivation system with types in which trust is explicitly formulated as a functional operator on typed terms and in which meta-theoretical results are interpreted accordingly. In particular, we present safety intended as the ability to prove whether the system under consideration produces nothing beyond what it is intended and designed to do, within certain margins of certainty due to its probabilistic nature. Accordingly, two main areas of research are surveyed in the following.

The extension of inferential systems based on  $\lambda$-calculi or type systems  with probabilities has been recently explored in the literature. The interest is mainly motivated by the need to develop semantics for probabilistic processes and programs, and to perform inferences on them. There are examples of stochastic untyped $\lambda$-calculi which focus on denotational or operational semantics with random variables, see respectively \citep{DBLP:conf/lics/BacciFKMPS18, DBLP:conf/icfp/BorgstromLGS16} and the uniform formulation for both offered in \citep{DBLP:conf/lics/AmorimKMPR21}. In this first sense, the meaning of a probabilistic program consists in enriching the language of programs with the ability to model  sampling from distributions, thereby rendering program evaluation a probabilistic process. In this way, the work in \citep{DBLP:conf/icfp/BorgstromLGS16}  introduces a measure space on terms and defines step-indexed approximations. A sampling-based semantics of a term is then defined as a function from a trace of random samples to a value. Differently, our work while also using expressions denoting values from samples, associates the probabilistic approximation with the value in terms of types, while indexing terms with the sample size. In this sense, obviously, the expressive power of our typed language is closer to \citep{DBLP:journals/pacmpl/DahlqvistK20}, which is designed specifically to model Bayesian learning; nonetheless, their semantic explanation of expressions in the language is significantly different, as expressions in our system denote samplings which generate outputs with a certain frequency, assuming a certain theoretical probability distribution. 
Hence, the aim of TPTND and its inferential structure are in fact closer to works introducing some form of probability in calculi with types or natural deduction systems, in particular see \citep{pierro2020atypetheory,bor16, bor17,borivcic2019sequent,ghilezan2018probabilistic}.
\citep{pierro2020atypetheory} introduces a $\lambda$-calculus augmented with special ``probabilistic choice'' constructs, i.e., terms of the form $M = \{ p_1 M_1, \dots, p_n M_n \}$, meaning that term $M$ has probability $p_1, \dots, p_n$ of reducing to one of the terms $M_1, \dots, M_n$ respectively. Unlike TPTND, \citep{pierro2020atypetheory} deals with judgements that do not have a context and uses a subtyping relation for the term reduction. \citep{bor17} introduces ``probabilistic sequents'' of the form $\Gamma \vdash^n \Delta$, which are interpreted as stating that the probability of $\Gamma \vdash \Delta$ is greater than, or equal to, a function of the natural number $n$, that is, $1-n\varepsilon$ for a fixed threshold $\varepsilon$. A similar formalism is introduced in \citep{borivcic2019sequent} which, closer to the natural deduction in \citep{bor16}, formalises probabilistic reasoning through sequents of the form $\Gamma \vdash_{a}^{b} \Delta$, which are interpreted as empirical statements of the form ``the probability of $\Gamma \vdash \Delta$ lies in the interval $[a,b]$''.
Differently from this work, TPTND does not express explicitly probabilistic intervals, but only sharp probability values, while at the same time expressing such probability on output types, rather than on the derivability relation. Finally, \citep{ghilezan2018probabilistic} introduces the logic P$\Lambda_{\rightarrow}$ where it is possible to mix ``basic'' and ``probabilistic'' formulas, the former being similar to standard Boolean formulas, and the latter being formed starting from the concept of a ``probabilistic operator'' $P_{\geq s} M:\sigma$ stating that the probability of $M:\sigma$ is equal to or greater than $s$. The semantics for P$\Lambda_{\rightarrow}$ is Kripke-like and its axiomatisation is infinitary. 

Most significantly for type-theoretical models of probabilistic reasoning, \citep{DBLP:conf/types/Adams015} introduces a quantitative logic with fuzzy predicates and conditioning of states. The computation rules of the system can be used for calculating conditional probabilities in two well-known examples of Bayesian reasoning in (graphical) models. In the same family, \citep{DBLP:journals/corr/Warrell16} offers a Probabilistic Dependent Type System (PDTS) via a functional language based on a subsystem of intuitionistic type theory including dependent sums and products, expanded to include stochastic functions. It provides a sampling-based semantics for the language based on non-deterministic $\beta$-reduction. A probabilistic logic from PDTS introduced as a direct result of the Curry-Howard isomorphism is derived, and shown to provide a universal representation for finite discrete distributions.

TPTND offers an integrated way to deal with probability intervals (similarly to \citep{bor16}) and probabilistic choices (similarly to \citep{pierro2020atypetheory}). However, unlike these languages, TPTND was mainly designed to reason about and derive trustworthiness properties of computational processes. In fact, differently from all these other systems, TPTND has an explicit syntax both to deal with the number of experiments, and to prove trust in a process whenever the empirically verified probability is close enough to the theoretical one. 

The formulation of the syntax of TPTND for the evaluation of trustworthiness of probabilistic computational processes recalls an entirely different literature, namely that of logical systems in which trust is explicitly formulated as an operator in the language. There are several logical frameworks providing an interpretation of one of the many aspects of computational trust, from manipulation to verification to assessment \citep{DBLP:conf/itrust/Demolombe04,DBLP:conf/atal/Singh11a,DBLP:conf/ant/DrawelBS17,DBLP:journals/tomacs/Aldini18, DBLP:conf/esorics/AldiniT19,DBLP:journals/fgcs/DrawelQBS20}. Modal logics, and in particular knowledge and belief logics,  have been extensively investigated for trust, see, e.g., \citep{DBLP:journals/ai/Liau03,DBLP:journals/igpl/HerzigLHV10,DBLP:conf/prima/LiuL17}. As mentioned above, much less investigated are proof systems in which trust is formulated explicitly as an operator or a functional expression. An exception is the family of logics \texttt{(un)SecureND} whose most complete formulation, including also a relational semantics, is offered in \citep{DBLP:journals/jancl/Primiero20}. The proof-theoretic fragment of the logic has been applied to a number of areas, ranging from  software management \citep{DBLP:conf/pst/BoenderPR15,DBLP:conf/ifiptm/PrimieroB17, DBLP:journals/wias/PrimieroB18} to  modelling a trust and reputation protocol for VANET \citep{DBLP:conf/eurosp/PrimieroRCN17}; and the evaluation of the trustworthiness of online sources \citep{DBLP:conf/ifiptm/CeolinP19,DBLP:conf/goodit/CeolinDP21}. This proof-theoretical language, in its deterministic version also includes a Coq implementation (\url{https://github.com/gprimiero/SecureNDC}) for proof-checking. The system TPTND presented in this paper can be seen as the probabilistic extension of the fragment of \texttt{(un)SecureND} including only the closure of trust under negation which corresponds to distrust in the latter system. The present version of TPTND significantly extends and improves on the version previously formulated in \citep{DBLP:conf/atal/DAsaroP21}.

\section{TPTND}\label{sec:system}
{ We now introduce the system TPTND, which comprises a logical language that enables us to talk about probabilistic computational processes and related trustworthiness properties, and a proof system providing a set of rules to formally reason about these processes and establish whether they enjoy the trustworthiness properties expressible in the language. Before introducing the syntax of the system, let us present a procedural semantics that determines the kind of computational objects and phenomena that the language of TPTND is meant to describe.
}

\subsection{Computational Semantics}\label{sec:compsec}

{ The computational semantics of TPTND is a procedural semantics, that is, a semantics that determines the objects of discourse by way of a definition of their behaviour in procedural terms. As already mentioned, the semantics formally defines the objects and phenomena that the language of TPTND is supposed to describe.
}

{ In order to introduce the semantics, let us begin with a simple example  that clarifies which kind of computational and probabilistic phenomena we are interested in and how we formalise them semantically.

\begin{example}
Suppose that we are interested in the probabilistic behaviour of two coins $u$ and $v$. Suppose, moreover, that $u$ is a fair coin while $v$ returns heads (value that we denote by the type $\eta$) with probability $\frac{2}{3}$ and tails (value that we denote by the type $\tau$)  with probability $\frac{1}{3}$.

In our semantics, the coin $u$ will be defined by the two following rules:
\[\infer{\listerm \mapsto_{\frac{1}{2}} \listerm , u:\eta}{}\qquad \infer{\listerm \mapsto_{\frac{1}{2}} \listerm ,  u:\tau}{}\]
indicating that, in any environment $\listerm$, we can always observe with probability $\frac{1}{2}$ that the coin $u$ returns a value $\eta$ (for heads) and with probability $\frac{1}{2}$ that $u$ returns a value $\tau$ (for tails). 
Technically, what we called an environment is simply a list of typed terms, either like $u:\eta$ and $u:\tau$---thus indicating that we observed that a certain term $u$ returned a certain value $\eta$ or $\tau$---or of the form $t_n:\alpha_a$---thus indicating that the frequency with which an output of type $\alpha$ has been produced by a process $t$ during a series of $n$ executions is $a$. Similarly, the coin $v$ is defined by thefollowing two rules:
\[\infer{\listerm \mapsto_{\frac{2}{3}} \listerm , v:\eta}{}\qquad \infer{\listerm \mapsto_{\frac{1}{3}} \listerm ,  v:\tau}{}\]
These four rules define all we know---and care---about the two coins. By using them, we can represent experiments on the two coins. For instance, without any additional assumption, we can produce the following sequence of events:  
\[\mapsto_{\frac{1}{2}} \; u:\eta \;\mapsto_{\frac{1}{2}} \;u:\eta , u:\eta \;\mapsto_{\frac{1}{2}} \; u:\eta , u:\eta , u:\tau \; \mapsto_{\frac{1}{3}} \; u:\eta , u:\eta , u:\tau , v:\tau\]representing the case in which we throw three times $u$ and one time $v$ and we obtain twice heads from $u$, one time tails from $u$ and one time tails from $v$. Notice that the reduction arrows $\mapsto$ are labelled by the probability of the observation that we add to the list of terms preceding the arrow in order to obtain the list following the arrow. 

Now that we have represented the process of collecting several observations on the two coins, we can gather them in a statement about what happened to one coin:\[ \; u:\eta , u:\eta , u:\tau , v:\tau\; \mapsto_1 \; u_3:\eta_{\frac{2}{3}} , v:\tau\]by the sampling$^{\mapsto}$ rule that we will define below. After the sampling, our environment contains one term indicating that $u$ returned heads twice over three throws and one indicating that $v$ returned tails once. The reduction obtained by using this rule has probability $1$ since it does not represent a probabilistic event but simply the action of putting together several observations.

Suppose now that we keep throwing the dice $v$ and that we obtain heads two times. We can collect our observations on the value tails obtained from $v$ along with our observations on the value heads obtained from $u$---in both cases, we represent this by an application of the sampling$^{\mapsto}$ rule. What we obtain is the following list of terms: $u_3:\eta_{\frac{2}{3}} , v_3:\tau_{\frac{1}{3}}$. At this point, we can talk about the behaviour of both  coins together by further putting together our knowledge into a term $\langle u_3:\eta_{\frac{2}{3}} , v_3:\tau_{\frac{1}{3}} \rangle$ describing the probabilistic behaviour of the pair of coins with respect to the result of throwing them that consists in obtaining once heads and once tails (represented by the conjunction type $\eta \times \tau$):
\[u_3:\eta_{\frac{2}{3}} , v_3:\tau_{\frac{1}{3}}\;\mapsto_1\; \langle u_3:\eta_{\frac{2}{3}} , v_3:\tau_{\frac{1}{3}} \rangle : (\eta \times \tau)_{\frac{2}{3}\cdot \frac{1}{3}} \]As we will see later, other logical operations can be employed to study the behaviour of processes---for instance, considering how frequently a process yields one of two values or considering how frequently a process yields a value if another, possibly related process yields a value with a particular frequency.
\end{example}

}
 
We now formally define the set of rules of the operational semantics. Technically, the presented calculus is meant as a proof-system for the deduction of judgements expressing the probability of obtaining certain outputs (types) from possibly opaque, probabilistic programs (terms). Therefore, all there is to specify about the nature and behaviour of the elements of the domain of discourse is their computational behaviour.  
These elements are, in particular, distinct computational processes yielding outputs of mutually exclusive types. Hence, the event of a process $t$ yielding an output $\alpha$---formally, $t:\alpha$---is always supposed to be stochastically independent from the event of another process $u$ yielding another output $\beta$ {\it and} from
the event of another process $u$ yielding the same output $\alpha$. Hence, providing a semantics in terms of generic events and their probabilities would not be very interesting and is certainly not required to fully understand the workings of the proof-system at hand. The calculus is, indeed, {\it not} supposed to be a means to compute the probability of generic events or processes. 
We focus then on the reduction rules determining the computational behaviour of our processes, which are presented in Figures \ref{fig:eval} and \ref{fig:eval2}, and discuss their meaning.

While the syntax of terms will be explicitly specified in Definition \ref{def:syntax} along with the rest of the syntax of the calculus TPTND, for the present section, as per usual practice, we just define a well-formed typed term as any element of a list $\listerm '$ occurring inside an expression of the form $\listerm\mapsto\listerm '$ that can be obtained by applying one or more of the rules in Figures \ref{fig:eval} and \ref{fig:eval2}. In other words, an evaluation rule acts on a list of typed terms (which, in the case of $\textrm{event}^{\mapsto t}$, might be empty) and produces a non-empty list of typed terms. Typed terms, therefore, can be of two forms:
\begin{itemize}
\item $t:\alpha$, denoting that the process $t$ produced an output of type $\alpha$; and
\item $t_n:\alpha_f$, denoting that $f$ is the frequency with which we obtained outputs of type $\alpha$ during an experiment consisting of $n$ executions of $t$---in other words, this denotes that the process $t$ produced $f\cdot n$ times an output of type $\alpha$ over a total of $n$ executions.
\end{itemize}
For each {\it atomic} process $t$  that we wish to study, we should associate a list $\alpha ^1, \ldots ,\alpha ^m$ of atomic types to it and we should define an $\textrm{event}^{\mapsto t}$  rule. Such a rule will define the probabilistic computational behaviour of $t$. Evaluating a list of typed terms will correspond, then, to collecting the results of experiments  (see Figure \ref{fig:eval}) and of logical operations (see Figure \ref{fig:eval2}) executed on (possibly several copies of) the probabilistic computational processes that we have introduced by the $\textrm{event}^{\mapsto t}$ rules.

We discuss now the specific meaning of each evaluation rule. We begin with the rules for evaluating atomic terms and then we consider the rules implementing logical operations on typed terms indicating the frequency of outputs of a certain type during experiments on probabilistic terms.

\begin{figure}

\[\infer[\textrm{event}^{\mapsto t}]{\listerm \quad \mapsto_{a_i}\quad \listerm, t:\alpha^i  }{}\]where each $a_i$ is a number in the unit interval $[0,1]$, $t$ is an atomic term, $\alpha^1, \dots , \alpha ^m$ are the atomic types  associated to $t$, $1\leq i\leq m $, $(\sum^{m}_{i=1} a_i)=1$
\[\infer[\textrm{sampling}^{\mapsto}]{\listerm, t:\alpha^1 , \dots , t:\alpha^n \quad \mapsto_{1}\quad \listerm, t_{n}:\alpha_f}{}\]
where $f=\frac{\mid\{i\,\mid\, \alpha _i=\alpha\}\mid }{n}$ 
\[\infer[\textrm{update}^{\mapsto}] {\listerm, t_n:\alpha_f ,t_m :\alpha_g \quad \mapsto_{1}\quad \listerm, t_{n+m}:\alpha_{f\cdot(n/(n+m))+g\cdot (m/(n+m))}}{}  \]
\caption{Evaluation Rules for Atomic Terms}
\label{fig:eval}
\end{figure}As already mentioned, the $\textrm{event}^{\mapsto t}$ rule defines the behaviour of the atomic term $t$. Technically, the rule enables us to add to our possibly empty list $\listerm$ a typed term indicating the outcome of one execution of $t$. Since $t$ denotes a probabilistic process, the rule will add to the list a typed term of the form $t:\alpha^i$ which specifies that $t$ produced an output of type $\alpha ^i\in \{\alpha ^1, \dots , \alpha ^m\}$. The choice of a specific type $\alpha ^i$ is non-deterministic, and the probability of selecting the type $\alpha ^i$ is $a_i$, for each $i\in \{1, \dots , m \}$. The conditions on this evaluation rule guarantee that the sum of the probabilities of all possible types of outputs of a term $t$ is always $1$.

 The sampling rule enables us to collect the results of $n$ executions of an atomic term $t$ into one typed term $t_n:\alpha _f$ indicating the frequency of the outputs of a type $\alpha$. In particular, the term $t_n:\alpha _f$ indicates the number of times an output of type $\alpha $ has been produced over the total number of outputs produced during the $n$ executions of $t$, i.e., $f\cdot n$ will be the number of times $t$ produced an output of type $\alpha$ during the experiment consisting of $n$ executions of $t$. 

The update rule combines, by a Bayesian update function, two typed terms $t_n:\alpha_f $ and $t_m :\alpha_g$ indicating the frequency with which an output of type $\alpha$ has been produced by $t$ during, respectively, an experiment consisting of $n$ executions of $t$ and an experiment consisting of $m$ executions of $t$, and produces a typed term $t_{n+m}:\alpha_{f\cdot(n/(n+m))+g\cdot (m/(n+m))}$ indicating the frequency with which an output of type $\alpha$ has been produced by $t$ during an experiment consisting of $n+m$ executions of $t$.

 We explain now the meaning of the logical operation rules, presented in Figure \ref{fig:eval2}.

\noindent
\begin{figure}

\[\infer[\mathrm{I}^{\mapsto}+]{\listerm, t_n:\alpha_f ,t_n :\beta_g \quad \mapsto_{1}\quad \listerm,  t_n :(\alpha +\beta )_{f+g}}{}\]

\[\infer[\mathrm{E}^{\mapsto}+_L]{\listerm, t_n :(\alpha +\beta ) _{f}  \quad \mapsto_{1}\quad  \listerm, t_n :\alpha_{f-g}  }{t_n : \beta_g }\]

\[\infer[\mathrm{E}^{\mapsto}+_R]{\listerm, t_n :(\alpha +\beta ) _{f}  \quad \mapsto_{1}\quad  \listerm, t_n :\beta_{f-g}  }{t_n : \alpha_g }\]

\[\infer[\mathrm{I}^{\mapsto}\times]{\listerm, t_n:\alpha_f ,u_n :\beta_g\quad \mapsto_{1}\quad \listerm, \langle t,u \rangle_{n} :(\alpha \times\beta )_{f\cdot g }}{}  \]


\[\infer[\mathrm{E}^{\mapsto}\times_L]{\listerm, t_n :(\alpha \times\beta ) _{f}  \quad \mapsto_{1}\quad  \listerm, \textit{fst}(t)_{n} :\alpha _{f/g}  }{\textit{snd}(t)_n : \beta_g }  \]

\[\infer[\mathrm{E}^{\mapsto}\times_R]{\listerm, t_n :(\alpha \times\beta ) _{f}  \quad \mapsto_{1}\quad  \mathcal{L}, \textit{snd}(t)_{n} :\beta_{f/g}  }{\textit{fst}(t)_n : \alpha _g }\]



\[\infer[\mathrm{I}^{\mapsto}\rightarrow]{\mathcal{L},t_n :\beta_f\quad \mapsto_{1}\quad \mathcal{L},  [x_u]t_n:(\alpha\rightarrow\beta)_{[a]f}}{x_u:\alpha_a }  \]

\[\infer[\mathrm{E}^{\mapsto}\rightarrow]{\mathcal{L}, [x_u]t_n:(\alpha\rightarrow\beta)_{[a]f} , u_n : \alpha _g \quad \mapsto_{1}\quad  \mathcal{L}, t_n.[u_n:\alpha ]:\beta_{g\cdot f}  }{x_u:\alpha _a}\]

where $\alpha$ and $\beta$ are metavariables for generic (possibly non-atomic) types, we assume, for $\mathrm{E}^{\mapsto}+_L, \mathrm{E}^{\mapsto}+_R, \mathrm{E}^{\mapsto}\times_L, \mathrm{E}^{\mapsto}\times_R$, that $\alpha \neq \beta$, and $x_u$ is a designated variable that we associate to the term $u$

	\caption{Logical Term Evaluation Rules}\label{fig:eval2}
\end{figure}

The rule I$^{\mapsto}$+ enables us to combine two typed terms specifying the frequency of the outputs of type $\alpha $ during an experiment consisting of $n$ executions of the process $t$ and of the outputs of type  $\beta$ during an experiment consisting of $n$ executions of $t$, in order to indicate the frequency of outputs of type either $\alpha $ or $\beta$ (namely of outputs of type $\alpha +\beta$) during an experiment consisting of $n$ executions of $t$.

The rules E$^{\mapsto}$+ can be applied to a typed term indicating the frequency of outputs of type $\alpha+\beta$, with respect to  an experiment consisting of $n$ executions of $t$, in order to extract the information about the frequencies of the outputs of type $\alpha$ and of the outputs of type $\beta$. If we wish to extract the information about the frequencies of the outputs of type $\alpha$ we use E$^{\mapsto}$+$_L$ and we need the information on the outputs of type $\beta$ indicated in the premise of the rule. In order to extract the information about the frequencies of the outputs of type $\beta$ we use E$^{\mapsto}$+$_R$ and we need the information on the outputs of type $\alpha$ indicated in the premise of the rule.

The rule I$^{\mapsto}\times$ enables us to combine two typed terms specifying the frequency of the outputs of type $\alpha $ during an experiment consisting of $n$ executions of the process $t$ and of the outputs of type  $\beta$ during an experiment consisting of $n$ executions of the process $u$, in order to indicate the frequency of outputs of type $\alpha\times\beta$ (namely of pairs of outputs in which the first element is of type $\alpha$ and the second of type $\beta$) during an experiment consisting of $n$ executions of the pair of processes $\langle t,u \rangle$. 


The rules  E$^{\mapsto}\times$ can be applied to a typed term indicating the frequency of outputs of type $\alpha\times\beta$ (and thus of pairs of outputs in which the first element is of type $\alpha$ and the second of type $\beta$), with respect to an experiment consisting of $n$ executions of $t$, in order to extract the information about the frequencies of the outputs of type $\alpha$ and of the outputs of type $\beta$ occurring as first and second element, respectively, of the considered output pairs. If we wish to extract the information about the frequencies of the first elements of the pairs (those of type $\alpha$) we use E$^{\mapsto}\times_L$ and we need the information on the second elements of the considered pairs of outputs, expressed by the premise of the rule. In order to extract the information about the frequencies of the second elements of the pairs (those of type $\beta$) we use E$^{\mapsto}\times_R$ and we need the information on the first elements of the considered pairs of outputs, expressed by the premise of the rule.

The rule I$^{\mapsto}\rightarrow$ enables us to make explicit a dependency between the frequency of the outputs of type $\beta$ of a process $t$ and the probability of obtaining an output of type $\alpha$ from a process $u$. The rule yields a term $[x]t$ indicating the dependency of the process $t$ on the value variable $x$ indicating the theoretical probability of obtaining $\alpha$. The type of the term $[x]t$ is $\alpha\rightarrow \beta $ in order to  indicate the existence of a dependency also at the level of types.

The rule E$^{\mapsto}\rightarrow$ enables us to combine a typed term  $[x_u]t_n:(\alpha\rightarrow\beta)_{[a]f}$ with a typed term $u_n:\alpha _g$. Since the first term  $[x_u]t_n:(\alpha\rightarrow\beta)_{[a]f}$ indicates that the frequency $f$ of the outputs of type $\beta$ produced by a process $t$ depends on the probability $a$ of obtaining an output of type $\alpha$ from a process $u$ and since the second term $u_n:\alpha _g$ indicates that an experiment of $n$ executions of the process $u$ yielded outputs of type $\alpha$ with a frequency of $g$, the resulting term $t_n.[u_n:\alpha ]:\beta_{g\cdot f}$, obtained by applying the first term to the second, indicates that the frequency of outputs of type $\beta$ obtained during an experiment consisting of $n$ executions of $t$ is $f\cdot g$ since the frequency of this type of output depends on the, now known, frequency $g$ of the outputs of type $\alpha $ produced by $u$.

\subsection{Syntax}\label{subsec:syntax}
The syntax of TPTND is generated by the following grammar in Backus-Naur form:

\begin{definition}[Syntax]\label{def:syntax}
	\begin{displaymath}
	\begin{array}{l}
	\mathtt{X} := \mathtt{x} \mid \mathtt{x}_\mathtt{T} \mid \langle \mathtt{X}, \mathtt{X} \rangle \mid \fst{(\mathtt{X})} \mid \snd{(\mathtt{X})} \mid [\mathtt{X}]\mathtt{X} \mid \mathtt{X}.\mathtt{X} \mid \mathtt{X}.\mathtt{T} \\

	\mathtt{T} := \mathtt{t} \mid \langle \mathtt{T},\mathtt{T} \rangle \mid \fst(\mathtt{T}) \mid \snd(\mathtt{T}) \mid
	[\mathtt{X}]\mathtt{T}\mid
	\mathtt{T}.\mathtt{T}\mid
	\Trust(\mathtt{T})\mid \UTrust(\mathtt{T})\\

	\mathtt{O}:= \alpha \mid \mathtt{O}_r \mid \mathtt{O}^\bot_{r} \mid  (\mathtt{O} \times \mathtt{O})_r \mid (\mathtt{O} + \mathtt{O})_r \mid (\mathtt{O} \rightarrow \mathtt{O})_{[r']r}\\
	
	\mathtt{C}:= \{ \} \mid \mathtt{C}, \mathtt{X} : \mathtt{C} \mid \mathtt{C}, \mathtt{X} : \mathtt{C}_r\mid \mathtt{C}, \mathtt{X} : \mathtt{C}_{[r,r']}\\
	
	\mathtt{S}:= \mathtt{C}:: distribution \mid \mathtt{O}:: output \mid 
 r::\mathbb{R}

	\end{array}
	\end{displaymath}
\end{definition}
We denote \textit{random variables} by $\mathtt{X}$. The metavariable $\mathtt{x}$ denotes a particular random variable, e.g., $x, x', y, z$.
Indexed random variables $x_\mathtt{T}$ are tied to specific terms, meaning that they describe uncertainty about the output of the denoted process $t$. We may combine independent random variables into pairs of the form $\langle \mathtt{X}, \mathtt{X} \rangle$ which intuitively corresponds to their joint distribution. Pairs have associated projection functions $\fst(\mathtt{X})$ and $\snd(\mathtt{X})$ to extract the first and second variable. We use the constructions $[\mathtt{X}]\mathtt{X'}$ and $\mathtt{X}.\mathtt{X'}$ (with $\mathtt{X'}$ a metavariable for a possibly distinct variable) to express the dependency relation between random variables, i.e., the probability assigned to the random variable $\mathtt{X}$, given the probability assigned to $\mathtt{X'}$ and the corresponding evaluation. The construction $\mathtt{X}.\mathtt{T}$ expresses the update of the probability assigned to random variable $\mathtt{X}$ when using the information on the frequency of the output of a process $\mathtt{T}$.

\textit{Terms} $\mathtt{T}$ are executed computational processes or computational experiments. The metavariable $\mathtt{t}$ denotes any constant term, e.g., $t, t', u,w, \dots$. Terms with different names are regarded as \textit{distinct} processes. When \textit{one and the same} term for a process occurs in distinct branches of a derivation (possibly with different numerical subscripts), these must be intended as distinct runs or executions of the same process.
Sometimes we use superscripts $t^{1}, \dots, t^{n}$ to distinguish the first to the $n^\text{th}$ distinct executions of a process $t$. We may compose processes into pairs $\langle \mathtt{T},\mathtt{T} \rangle$ associated with projection functions $\fst(\mathtt{T})$ and $\snd(\mathtt{T})$. The construct $\mathtt{[X]T}$ expresses the expected probability assigned to the output of process $\mathtt{T}$ given the theoretical probability assigned to $\mathtt{X}$; while the construct $\mathtt{T}.\mathtt{T}$ expresses the update of the expected probability assigned to the output of a process $\mathtt{T}$ with the expected probability of the output of another process of the form $\mathtt{T}$, where the latter is taken to substitute the dependency of $\mathtt{T}$ from some random variable $\mathtt{X}$. We include among terms the function-like expressions for trust: $\Trust(\mathtt{T})$ expresses that the frequency of process $\mathtt{T}$ is considered trustworthy  with respect to the theoretical probability of its output, and the trust operator is defined by an appropriate pair of introduction and elimination rules below. The function $\UTrust(\mathtt{T})$ expresses negated trustworthiness for the frequency shown by the output of process $\mathtt{T}$ with respect to its intended theoretical probability.

Types $\mathtt{O}$ are output values of random variables and computational processes, with a probability $r\in [0,1]$ attached. Similarly as in the case of terms and random variables, $\alpha$ and $\beta$ are used as metavariables for generic types (thus possibly non atomic). Output $\alpha \times \beta$ expresses the independent occurrence of outputs $\alpha$ and $\beta$. Output $\alpha + \beta$ expresses the disjunction of outputs $\alpha$ and $\beta$. The construction $(\alpha \rightarrow \beta)_{[r']r}$ is interpreted as the probability $r$ of output $\beta$ for a given random variable or process, provided the probability of output $\alpha$ is $r'$. We use $\alpha^{\bot}$ to denote any output different than $\alpha$ (possibly including no output). 
We will use subscripts $a$, $b$, $c$, $\dots$ to denote instance values of probability $r$ when referring to the output of a random variable $\mathtt{X}$. 
Hence, in a formula of the form $x:\alpha_{a}$, the subscript $a$ denotes the probability associated with the event of random variable $x$ taking value $\alpha$. The formula $t_{n}:\alpha_{\tilde{a}}$ expresses the expected probability $a$ that process $t$ will output a value of type $\alpha$. In the expression $t_{n}:\alpha_{f}$, the subscript denotes the empirical probability or frequency associated with output $\alpha$ extracted from $n$ executions of the process denoted by $t$.  By convention, we use $t:\alpha$ without subscript either in the term or in the type to denote a single execution of the process $t$ with deterministic output $\alpha$ (i.e., $r=1$).

\textit{Contexts} $\mathtt{C}$ are sets of conditions, or assumptions on the theoretical probabilities of random variables required for the assertions made on the right of $\vdash$ to hold. Contexts are composed by lists of  expressions of the form $x_{t}:\alpha_{a}$, with or without probability $a$ attached. They presuppose type declarations $\alpha::output$ for the occurring outputs, saying that $\alpha$ is a valid output. We  assume that, even when the actual distribution of probabilities on the random variables of interest of a given system is unknown, a certain assumption on their values can be made, and in particular the assumption that the distribution is unknown can be stated.

\textit{Type declarations} $\mathtt{S}$ include statements of the form: $\mathtt{C}::distribution$, expressing that a given (possibly non empty) list of declarations of the form $x:\alpha_{a}$ is a probability distribution; $\alpha::output$ expressing that $\alpha$ is a valid output value (e.g., a Boolean, or a Natural, or a List, etc.);
$r::\mathbb{R}$ expressing that $r$ is a real number. 
These declarations are generally considered as implicit presuppositions of contexts.

\subsection{Judgements}

We now summarize the forms of judgements admissible in TPTND. A judgement of the form 

\[ x:\alpha \vdash y:\beta_b \]

says that the probability of random variable $y$ to have value $\beta$ is $b$, provided random variable $x$ has value $\alpha$. From the probabilistic point of view, this may be intuitively interpreted as $P(Y=\beta\mid X=\alpha)=b$. In the particular case of an empty context, we may simply interpret \[\vdash y:\beta_{b}\]as $P(Y=\beta)=b$. 

A judgement of the form 

\[ x:\alpha_a \vdash y:\beta_b \]

says that the probability of random variable $y$ to take value $\beta$ is $b$ if $x$ takes value $\alpha$ with probability $a$. In other words, the probability $ P(Y=\beta) $ is a function, say $g$, of $P(X=\alpha)$, and $g(a) = b$. 

A judgement of the form 

\[\vdash t:\alpha\]
says that an event or process $t$ has taken place with output value $\alpha$. 

A judgement of the form

\[ x_{t}:\alpha_a \vdash t:\beta \]

serves the purpose of reasoning about data while making assumptions on its probability distribution; it may be intuitively thought as saying that a process $t$ produces output $\beta$ assuming a variable $x$ has probability $a$ of producing output $\alpha$. This form of judgement is then generalized
and formulated as expected probabilities under some distribution:

\[ x_{t}:\alpha_a \vdash t_{n}:\beta_{\tilde{b}}\]

says that a process $t$ produces output $\beta$ with expected probability $\tilde{b}$ over $n$ executions assuming a variable $x$ has probability $a$ of producing output $\alpha$. 

Next, we can formulate this judgement under more assumptions on  probability values for conditions to be satisfied on the left:

\[ x_{t}:\alpha_{a}, \dots, z_{t}:\nu_{n} \vdash t_{n}:\beta_{\tilde{b}} \]

says that provided variables $x, \dots, z$ with probabilities $a, \dots, n$ to produce outputs $\alpha, \dots \nu$ respectively obtain, the associated process $t$ produces output $\beta$ with expected probability $\tilde{b}$.

The relative frequency of $\beta$ in $n$ executions of $t$ is denoted by

%
%
%
$$x_{t}:\alpha_{a}, \dots, z_{t}:\nu_{n}\vdash t_{n}:\beta_{f}$$
%

Note that as these are judgements based on assumptions, the evaluation of the frequency with which $\beta$ is shown (or its expected probability to occur) requires the evaluation of the conditions: given a judgement of the form $x_{t}:\alpha_a \vdash t:\beta$, the assertion that $t$ produces $\beta$ requires the judgement $\vdash x:\alpha$ to hold; the same is true for multiple assumptions. In the case of a judgement $x_{t}:\alpha_a \vdash t_{n}:\beta_{\tilde{b}}$, the verification of the condition can only be formulated over the same amount of executions as $t$, in which case the judgement $\vdash u_{n}:\alpha_{\tilde{a}}$ is required. 


Finally, a judgement of the form

$$x_{t}:\alpha_{[0,1]}, \dots, z_{t}:\nu_{[0,1]}\vdash t_{n}:\beta_{f}$$
says that $n$ executions of process $t$ display output $\beta$ with a frequency $f$ under an unknown distribution of values to random variables. Note that while no inference is possible from a judgement which is indexed by a range $[0,1]$ within the regular deductive rules of our system, we admit judgements of this form as a premise in the rule to infer trust or distrust in $t_{n}:\beta_{f}$ (see Section \ref{subsec:trust}), as well as in the premise to apply contraction, where all the values in the interval may be applied to an output value and a function is executed to extract one of them (see Section \ref{subsec:structural}).

\subsection{Distribution Construction Rules}

Contexts as a list of assumptions on probability distributions are inductively constructed, see Figure \ref{fig:distributions}. They are generated by judgements of the form $\vdash x:\alpha_{a}$, expressing assignments of probabilities to output variables. An empty set $\{\}$ represents the base case (rule ``base'') to define distributions. Random variables assigning theoretical probability values to outputs can be added to a distribution as long as they respect the standard additivity requirement on probabilities (``extend''). Here $\alpha$ is intended as a metavariable for both atomic and compound output types; the cases for the rule ``extend'' include joint distribution and extension by dependent variables, as justified below in the construction rules for random variables. The variant for deterministic outputs ``extend\_det'' has $x:\alpha\notin \Gamma$ in the second premise and it does not require the additivity condition. { Indeed, the meaning of $x:\alpha$ is that the random variable $x$ has been assigned the value $\alpha$, and this is not in contrast with an assumption of the form $x:\alpha_a$ that establishes a theoretical probability for that variable}. To denote an unknown distribution (``unknown'') we use the (transfinite) construction assigning to all values of interest the interval of all possible probability values, provided the additivity condition of ``extend'' is preserved. Note that, while the fragment of the calculus already introduced may not apply such unknown distribution to infer values on the outputs, as we will see in Section \ref{subsec:trust}, the Trust fragment (see Figure \ref{fig:TrustFragment}) admits the use of an unknown distribution in the second premise of the introduction rules (IT/IUT) and in the conclusion of the corresponding eliminations (ET/EUT). When an unknown distribution is used, the frequency values assigned to the outputs on the right-hand side of our judgement must be understood as observed values under an unknown distribution.

\begin{example}\label{ex:fairdie}
$$\Gamma=\{x_{d}:1_{1/6}, \dots,  x_{d}:6_{1/6}\}$$

is the theoretical probability distribution associated with the outputs of a fair die $d$.
\end{example}

\begin{example}
$$\Gamma=\{ x_{d}:\Heads_{1/2}, x_{d}:\Tails_{1/2}, y_{d'}:\Heads_{1/3}, y_{d'}:\Tails_{2/3} \}$$ 

is the joint theoretical distribution of two independent random variables $x_{d}$ and $y_{d'}$ associated with a fair coin $d$ and biased coin $d'$ respectively.
\end{example}

\begin{figure}
    \begin{prooftree}
        \AxiomC{}
        \RightLabel{base}
        \UnaryInfC{$\{\} :: \distribution $}
        \end{prooftree}

\begin{prooftree}
   \AxiomC{$\Gamma::\distribution$}
   \AxiomC{$x:\alpha_a \notin \Gamma, \forall a\in[0,1]$}
   \AxiomC{$(\sum_{x:\rho_p \in \Gamma}p)+a \leq 1  $}
   \RightLabel{extend}
   \TrinaryInfC{$\Gamma, x:\alpha_a :: \distribution$}
\end{prooftree}

\begin{prooftree}
   \AxiomC{$\Gamma::\distribution$}
   \AxiomC{$x:\alpha\notin \Gamma$}
   \RightLabel{extend\_det}
   \BinaryInfC{$\Gamma, x:\alpha :: \distribution$}
\end{prooftree}

\begin{prooftree}
   \AxiomC{$\alpha::output$}
   \AxiomC{$\dots$}
   \AxiomC{$\omega::output$}
   \RightLabel{unknown}
   \TrinaryInfC{$\{x:\alpha_{[0,1]}, \dots,  x:\omega_{[0,1]}\}:: \distribution$}
\end{prooftree}
        \caption{Distribution Construction}\label{fig:distributions}
\end{figure}

\subsection{Reasoning about random variables}\label{subsec:reason_random}

\begin{figure}
        \begin{prooftree}
        \AxiomC{}
        \RightLabel{identity$_1$}
        \UnaryInfC{$\Gamma, x:\alpha \vdash x:\alpha_1$}
        \end{prooftree}

        \begin{prooftree}
        \AxiomC{}
        \RightLabel{identity$_2$}
        \UnaryInfC{$\Gamma, x:\alpha_a \vdash x:\alpha_a$}
        \end{prooftree}

    \centering
        \begin{prooftree}
        \AxiomC{$\Gamma\vdash x:\alpha_a$}
        \RightLabel{$\bot$}
        \UnaryInfC{$\Gamma\vdash x:\alpha^{\bot}_{1-a}$}
    \end{prooftree}
    
    \begin{prooftree}
            \AxiomC{$\Gamma \vdash x:\alpha_a$}
            \AxiomC{$\Delta \vdash y:\beta_b$}
            \AxiomC{$\Gamma\indep\Delta$}
        \RightLabel{I$\times$}
            \TrinaryInfC{$\Gamma, \Delta \vdash \langle x,y\rangle : (\alpha\times\beta)_{a\cdot b}$}
    \end{prooftree}

        \begin{prooftree}
            \AxiomC{$\Gamma \vdash x : (\alpha\times\beta)_{c}$}
            \AxiomC{$\Gamma \vdash \fst (x) : \alpha_{a}$}
            \RightLabel{E$\times_L$}
            \BinaryInfC{$\Gamma \vdash \snd (x) : \beta_{c/a} $}
    \end{prooftree}
    
    \begin{prooftree}
            \AxiomC{$\Gamma \vdash x_n : (\alpha\times\beta)_{c}$}
            \AxiomC{$\Gamma \vdash \snd (x): \beta_{b}$}
            \RightLabel{E$\times_R$}
            \BinaryInfC{$\Gamma \vdash \fst (x) : \alpha_{c/b} $}
    \end{prooftree}

    \begin{prooftree}
            \AxiomC{$\Gamma \vdash x:\alpha_a$}
            \AxiomC{$\Gamma \vdash x:\beta_b$}
               \RightLabel{I$+$}
                \BinaryInfC{$\Gamma \vdash x : (\alpha+\beta)_{a+b}$}
    \end{prooftree}

        \begin{prooftree}
                \AxiomC{$\Gamma \vdash x : (\alpha + \beta)_{a} $}
                \AxiomC{$\Gamma \vdash x : \alpha_{a'} $}
                \RightLabel{E$+_R$}
                \BinaryInfC{$\Gamma \vdash x : \beta_{a-a'} $}
        \end{prooftree}
        
    \begin{prooftree}
                \AxiomC{$\Gamma \vdash x : (\alpha + \beta)_{b} $}
                \AxiomC{$\Gamma \vdash x : \beta_{b'} $}
                \RightLabel{E$+_L$}
                \BinaryInfC{$\Gamma \vdash x : \alpha_{b-b'} $}
        \end{prooftree}

    \begin{prooftree}
            \AxiomC{$\Gamma, x:\alpha \vdash y : \beta_{b}$}
            \RightLabel{I$\rightarrow$}
            \UnaryInfC{$ \Gamma \vdash [x]y:(\alpha\rightarrow\beta)_{b} $}
    \end{prooftree}

    \begin{prooftree}
          	\AxiomC{$\Gamma \vdash [x] y:(\alpha\rightarrow \beta)_{b}$}
          	\AxiomC{$\Gamma \vdash x: \alpha_{a}$}
         	\RightLabel{E$\rightarrow$}
          	\BinaryInfC{$\Gamma \vdash y.(x:\alpha):\beta_{a\cdot b}$}
  \end{prooftree}
    
where we assume, for $\mathrm{E}+_L, \mathrm{E}+_R, \mathrm{E}\times_L, \mathrm{E}\times_R$, that $\alpha \neq \beta$

\caption{Rules for Random Variables}\label{fig:variablesrules}
\end{figure}

The inferential engine of TPTND for random variables is expressed by the rules in Figure \ref{fig:variablesrules}.

The axiom-like rule ``identity$_1$'' expresses the probabilistic fact that $P(X=\alpha\mid X=\alpha)=1$.  The rule ``identity$_2$'' generalizes the former by expressing that the probability of a random variable $x$ to have output $\alpha$ is $a$, provided variable $x$ has theoretical probability $a$ to output $\alpha$.

The rule ``$\bot$'' encodes the probabilistic complement rule: if the probability of output $\alpha$ for the random variable $x$ is $a$, the probability of $x$ not producing output $\alpha$ is $1-a$.

The joint distribution of two independent random variables can be constructed with conjunction $\times$. We use the notation ``$\Gamma \indep \Delta $'' with the intended meaning that contexts $\Gamma$ and $\Delta$ are independent, meaning that for any $x:\alpha_{a}\in \Gamma$, $y:\beta_{b}\in \Delta$ distinct, is never the case that $\Delta\vdash x:\alpha_{a}$ and $\Gamma\vdash y:\beta_{b}$. Note that distinct variables can be made dependent through the use of $\rightarrow$ and that the same variable can occur in independent contexts. The rule generalizes the independence condition to entire distributions and since it assumes their independence, the joint distribution multiplies their probabilities. The elimination uses $\fst$ and $\snd$ to derive the probability of the other term by dividing probabilities.

Disjunction $+$ merges mutually exclusive outputs for a given variable by adding their probabilities. Elimination of disjunction amounts to subtracting them instead.

The connective $\rightarrow$ is used as a syntactical rewriting of dependent variables. The implication introduction constructs a variable of type $\alpha \rightarrow \beta$ where the probability $b$ of $y$ being $\beta$ is explicitly made dependent on $x$ being $\alpha$. In the elimination, we assign a specific probability $a$ to $x$ being $\alpha$ and construct the application term $x.(y:\alpha)$ which is assigned output $\beta$ with probability $a\cdot b$. 


	\begin{example}[Dependent variables] We want to model a distribution with two random variables $x$ and $y$ corresponding to two coins, where $x$ is fair. Variable $y$ always returns $\Heads$ when $x$ returns $\Tails$, and vice-versa:
			
			\begin{prooftree}
				\AxiomC{$x:\Heads \vdash y:\Tails_{1}$}
				\RightLabel{I$\rightarrow$}
				\UnaryInfC{$\vdash [x]y:\Heads\rightarrow\Tails_{1}$}
			\end{prooftree}
			\begin{prooftree}
				\AxiomC{$x:\Tails \vdash y:\Heads_{1}$}
				\RightLabel{I$\rightarrow$}
				\UnaryInfC{$\vdash [x]y:\Tails\rightarrow\Heads_{1}$}
			\end{prooftree}
			
			We show that $y$ is also fair by eliminating implication:
			\begin{prooftree}
				\AxiomC{$\vdash [x]y:\Heads\rightarrow\Tails_{1}$}
				\AxiomC{$\Gamma\vdash x:\Heads_{0.5}$}
				\RightLabel{E$\rightarrow$}
				\BinaryInfC{$\Gamma\vdash y.(x:\Heads) : \Tails_{1\cdot 0.5}$}
			\end{prooftree}
			
			%
		\end{example}

\subsection{Reasoning about events, samples and expected probabilities}\label{subsec:reason_processes}

The inferential engine for single executions of computational processes performed under a given probability distribution is illustrated in Figure \ref{fig:singleExperimentRules}.

\begin{figure}

       \begin{prooftree}
       \AxiomC{}
       \RightLabel{experiment}
       \UnaryInfC{$\Gamma, x_{t}:\alpha_a \vdash t : \alpha$}
   \end{prooftree}

     \begin{prooftree}
            \AxiomC{$\Gamma \vdash t : \alpha$}
            \AxiomC{$\Delta \vdash u : \beta$}
            \AxiomC{$\Gamma \indep \Delta$}
            \RightLabel{I$\times$}
            \TrinaryInfC{$\Gamma, \Delta \vdash \langle t, u \rangle : (\alpha \times \beta) $}
    \end{prooftree}
    
        \begin{prooftree}
                \AxiomC{$\Gamma \vdash t : (\alpha \times \beta) $}
                \RightLabel{E$\times_L$}
                \UnaryInfC{$\Gamma \vdash \fst(t) : \alpha $}
    \DisplayProof
                \AxiomC{$\Gamma \vdash t : (\alpha \times \beta) $}
                \RightLabel{E$\times_R$}
                \UnaryInfC{$\Gamma \vdash \snd(t) : \beta $}
        \end{prooftree}

 {
    \begin{prooftree}
            \AxiomC{$\Gamma
            \vdash t : \alpha$}
            \RightLabel{I$+_L$}
            \UnaryInfC{$\Gamma \vdash t : (\alpha + \beta) $}
\DisplayProof
\quad 
            \AxiomC{$\Gamma
            \vdash t : \beta$}
            \RightLabel{I$+_R$}
            \UnaryInfC{$\Gamma \vdash t : (\alpha + \beta) $}
    \end{prooftree}\begin{center}
with $x_{t}:\alpha_{a}, x_{t}:\beta_{b}\in \Gamma$
    \end{center}
    }
    \begin{prooftree}
                \AxiomC{$\Gamma \vdash t : (\alpha + \beta) $}
                \AxiomC{$\Gamma \vdash t : (\alpha^{\bot}) $}
                \RightLabel{E$+_R$}
                \BinaryInfC{$\Gamma \vdash t : \beta $}
        \end{prooftree}
        
    \begin{prooftree}
                \AxiomC{$\Gamma \vdash t : (\alpha + \beta) $}
                \AxiomC{$\Gamma \vdash t : (\beta^{\bot}) $}
                \RightLabel{E$+_L$}
                \BinaryInfC{$\Gamma \vdash t : \alpha $}
        \end{prooftree}

 \begin{prooftree}
 \AxiomC{$\Gamma, x:\alpha_{a}\vdash t:\beta$}
 \RightLabel{I$\rightarrow$}
 \UnaryInfC{$\Gamma\vdash [x]t:(\alpha\rightarrow \beta)_{a}$}
\DisplayProof
\quad
 \AxiomC{$\Gamma\vdash [x]t:(\alpha\rightarrow \beta)_{a}$}
 \AxiomC{$\Delta\vdash u:\alpha$}
 \RightLabel{E$\rightarrow$}
 \BinaryInfC{$\Gamma , \Delta\vdash t.[u:\alpha]:\beta$}
 \end{prooftree}
        
where we assume, for $\mathrm{E}+_L, \mathrm{E}+_R$, that $\alpha \neq \beta$ and $x_t$ is a designated variable that we associate to the term $t$.
        
     \caption{Single-experiment rules.}
    \label{fig:singleExperimentRules}
\end{figure}

The rule ``experiment'' is the declaration of the deterministic output $\alpha$ (without subscript) of one experiment or execution of process $t$ (without subscript) under a random variable which assigns theoretical probability $a$ to output $\alpha$ (and hence derives $\alpha$ with that probability). Experiments can then be combined for dependent and independent results, in the same way done above for random variables.
{
We remark that the added condition $x_t:\alpha_a , x_t :\beta _b \in \Gamma$ to the I$+_L$ and I$+_r$ rules is needed to guarantee that, when $\alpha +\beta $ is inferred from $\alpha$ (respectively from $\beta$), $\beta $ (respectively $\alpha$) is indeed a possible output of $t$. Notice moreover that, given any term $t$, we associate to it a designated variable $x_t$, as indicated in Figure \ref{fig:singleExperimentRules}.} 

Inferences on processes which assume some probability value for a random variable are possible under the construction of the $\rightarrow$ connective: the introduction states that if process $t$ outputs value $\beta$ (with $r=1$) provided the condition that $x$ is assigned value $\alpha$ is satisfied with probability $a$, then one constructs term $[x]t$ which assigns the probability $a\cdot 1$ to the output $(\alpha \rightarrow \beta)$; the elimination says that on condition that a process $u$ with value $\alpha$ obtains ($r=1$), process $t$ using information $u$ outputs value $\beta$ ($r=1$), thereby eliminating the probabilistic condition.

\begin{figure}

   \begin{prooftree}
       \AxiomC{}
       \RightLabel{expectation}
       \UnaryInfC{$x_{t}:\alpha_{a}\vdash t_{n}:\alpha_{\tilde{a}}$}
\end{prooftree}       
   
   \begin{prooftree}
       \AxiomC{$\Gamma\vdash t^1:\alpha^1$}
       \AxiomC{$\dots$}
       \AxiomC{$\Gamma \vdash t^{n}:\alpha^n$}
       \RightLabel{sampling}
    \TrinaryInfC{$\Gamma\vdash t_{n} : \alpha_{f}$}
  \end{prooftree}
  \begin{center}
  where $f = \frac{\mid \{ i \mid \alpha^i = \alpha \} \mid}{n}$
  \end{center}

 \begin{prooftree}
 \AxiomC{$\Gamma \vdash t_n:\alpha_{f}$}
 \AxiomC{$\Gamma \vdash t_m:\alpha_{f'}$}
     \RightLabel{update}
     \BinaryInfC{$\Gamma \vdash t_{n+m} : \alpha_{f\cdot(n/(n+m))+f'\cdot(m/(n+m))} $}
      \end{prooftree}

      \begin{prooftree}
  
            \AxiomC{$\Gamma \vdash t_n: \alpha_{\tilde{a}}$}
            \AxiomC{$\Gamma \vdash t_n:\beta_{\tilde{b}}$}
            \RightLabel{I$+$}
            \BinaryInfC{$ \Gamma \vdash t_n : (\alpha + \beta)_{\tilde{a} + \tilde{b}} $}
    \end{prooftree}
    
    \begin{prooftree}
            \AxiomC{$\Gamma \vdash t_n : (\alpha + \beta)_{\tilde{c}}$}
            \AxiomC{$\Gamma \vdash t_n : \alpha_{\tilde{a}}$}
            \RightLabel{E$+_L$}
            \BinaryInfC{$\Gamma \vdash t_n : \beta_{\tilde{c}-\tilde{a}} $}
    \end{prooftree}
    
    \begin{prooftree}
            \AxiomC{$\Gamma \vdash t_n : (\alpha + \beta)_{\tilde{c}}$}
            \AxiomC{$\Gamma \vdash t_n : \beta_{\tilde{b}}$}
            \RightLabel{E$+_R$}
            \BinaryInfC{$\Gamma \vdash t_n : \beta_{\tilde{c}-\tilde{b}} $}
    \end{prooftree}

    \begin{prooftree}
            \AxiomC{$\Gamma \vdash t_n : \alpha_{\tilde{a}}$}
            \AxiomC{$\Delta \vdash u_n : \beta_{\tilde{b}}$}
             \AxiomC{$\Gamma \indep \Delta$}
            \RightLabel{I$\times$}
            \TrinaryInfC{$\Gamma, \Delta \vdash \langle t, u \rangle_{n} : (\alpha \times \beta)_{\widetilde{a \cdot b}} $}
    \end{prooftree}

         \begin{prooftree}
            \AxiomC{$\Gamma \vdash t_n : (\alpha\times\beta)_{\tilde{c}}$}
            \AxiomC{$\Gamma \vdash \fst (t)_n : \alpha_{\tilde{a}}$}
            \RightLabel{E$\times_L$}
            \BinaryInfC{$\Gamma \vdash \snd (t)_{n} : \beta_{\widetilde{c/a}} $}
    \end{prooftree}
    
    \begin{prooftree}
            \AxiomC{$\Gamma \vdash t_n : (\alpha\times\beta)_{\tilde{c}}$}
            \AxiomC{$\Gamma \vdash \snd (t)_n : \beta_{\tilde{b}}$}
            \RightLabel{E$\times_R$}
            \BinaryInfC{$\Gamma \vdash \fst (t)_{n} : \alpha_{\widetilde{c/b}} $}
    \end{prooftree}

    
    
\begin{prooftree}
\AxiomC{$\Gamma, x_{t}:\alpha_{a}\vdash t_{n}:\beta_{\tilde{b}}$}
\RightLabel{I$\rightarrow$}
\UnaryInfC{$\Gamma\vdash [x]t_{n}:(\alpha\rightarrow \beta)_{[a]\tilde{b}}$}
\end{prooftree}

\begin{prooftree}
\AxiomC{$\Gamma\vdash [x]t_{n}:(\alpha\rightarrow \beta)_{[a]\tilde{b}}$}
\AxiomC{$\Delta\vdash u_{n}:\alpha_{\tilde{a}}$}
\RightLabel{E$\rightarrow$}
\BinaryInfC{$\Gamma , \Delta\vdash t_{n}.[u_{n}:\alpha]:\beta_{\tilde{a}\cdot \tilde{b}}$}
\end{prooftree}

where we assume, for $\mathrm{E}+_L, \mathrm{E}+_R$, that $\alpha \neq \beta$ and $x_t$ is a designated variable that we associate to the term $t$.

 \caption{Expected Probabilities and Sampling Rules}\label{fig:connectiverules}
\end{figure}

Rules from Figure \ref{fig:singleExperimentRules} allow to perform complex experiments. We now want to generalize to multiple executions of an experiment in Figure \ref{fig:connectiverules}. This means to evaluate two things: first, the expected value of the probability of an output given $n$ executions of an experiment; second, the frequency of a given output in $n$ executions of an experiment. As expected probabilities are computed according to their theoretical counterparts, we also define logical operations on them.

The axiom-like rule ``expectation'' says that given the theoretical probability $a$
that the random variable associated with process $t$ will output value $\alpha$, the expected probability of having $\alpha$ over $n$ executions of $t$ is $\tilde{a}$. Provided the maximally expected probability of having $\alpha$ over $n$ executions is exactly $a$, we use $\tilde{a}$ as a notational device to distinguish theoretical and expected probabilities. 

In the rule ``sampling'' we denote a frequency value $f$ of output $\alpha$ from $n$ actual executions of process $t$, each with its own deterministic output, where $f$ is the number of cases in which $\alpha$ has occurred as output over the total number  of cases $n$. As a general case, the distribution under which $t$ is executed can be taken to be unknown. An experiment can be repeated several times for a given process $t$, and the rule ``update'' tells how to calculate relative frequencies on the various number of executions, each provided by an instance of ``sampling''. Note that in the examples below to aid readability we will sometimes formulate the rule ``sampling'' without appropriate premises (i.e., the corresponding several instances of the rule ``experiment''), but will always assume those are properly executed. Moreover, ``update'' is obviously a shorthand for an extended ``sampling'', but it is notationally useful when its premises occur at distinct stages of a derivation tree.

\begin{example}

Consider a die of which we know nothing, and in particular it is unknown whether it is fair or not. We denote this state of affairs with the distribution $\Gamma=\{x:1_{[0,1]}, \dots, x:6_{[0,1]}\}$, as produced by the type-definition ``unknown''. We execute two sets of four experiments to evaluate the frequency of output $1$.

\begin{prooftree}
            \AxiomC{$\Gamma\vdash d^{1}:1$}
            \AxiomC{$\Gamma\vdash d^{2}:1$}
            \AxiomC{$\Gamma\vdash d^{3}:5$}
            \AxiomC{$\Gamma\vdash d^{4}:6$}
            \RightLabel{sampling}
            \QuaternaryInfC{$\Gamma\vdash d_4:1_{1/2}$}
                  \end{prooftree}

\begin{prooftree}

            \AxiomC{$\Gamma\vdash d^{1}:3$}
            \AxiomC{$\Gamma\vdash d^{2}:1$}
            \AxiomC{$\Gamma\vdash d^{3}:5$}
            \AxiomC{$\Gamma\vdash d^{4}:6$}
            \RightLabel{sampling}
            \QuaternaryInfC{$\Gamma\vdash d_4:1_{1/4}$}
        \end{prooftree}

\begin{prooftree}
\AxiomC{$\mathcal{D}_{1}$}
\AxiomC{$\mathcal{D}_{2}$}
\RightLabel{update}
\BinaryInfC{$\{x:1_{[0,1]}, \dots, x:6_{[0,1]}\}\vdash d_{8} :1_{3/8}$}
        \end{prooftree}
\end{example}
\bigskip

Rule I$+$ introduces disjunction: intuitively, if under a distribution $\Gamma$ a sample of process $t$ produces output $\alpha$ with an expected probability $\widetilde{a}$, and output $\beta$ with expected probability $\widetilde{b}$, then the expected probability of output $\alpha$ or output $\beta$ by a run of process $t$ is $\widetilde{a}+\widetilde{b}$. Note that this rule can also be seen as an abbreviation of an instance of \textit{sampling}  from $n$ premises each of which is the result of introducing $+$ on a single experiment whose output is either $\alpha$ or $\beta$:

\begin{example}
The rule I$+$ followed by a sampling for a process with different outputs (including outputs possibly different from $\alpha$ or $\beta$):  
{\tiny 
   \begin{scprooftree}{0.80}   
     
     \AxiomC{$\Gamma, x_{t}:\alpha_{a}, x_{t}:\beta_{b}\vdash t^{1}:\alpha$}
     %
%
     \AxiomC{$\Gamma, x_{t}:\alpha_{a}, x_{t}:\beta_{b}\vdash t^{1}:\beta$}
\RightLabel{I$+$}
     \BinaryInfC{$\Gamma, x_{t}:\alpha_{a}, x_{t}:\beta_{b}\vdash t^{1}:(\alpha + \beta)$}

\AxiomC{$\dots$}
\UnaryInfC{$\dots$}

    \AxiomC{$\Gamma, x_{t}:\alpha_{a}, x_{t}:\beta_{b}\vdash t^{n}:\alpha$}
%
%
     \AxiomC{$\Gamma, x_{t}:\alpha_{a}, x_{t}:\beta_{b}\vdash t^{n}:\beta$}
\RightLabel{I$+$}
     \BinaryInfC{$\Gamma, x_{t}:\alpha_{a}, x_{t}:\beta_{b}\vdash t^{n}:(\alpha + \beta)$}

     
    \RightLabel{sampling}
    \TrinaryInfC{$\Gamma, x_{t}:\alpha_{a},x_{t}:\beta_{b}\vdash t_{n} : (\alpha+\beta)_{f}$}
  \end{scprooftree}
  }
\end{example}
\bigskip
Conversely by E$+_{R}$ (respectively E$+_{L}$): if under a distribution $\Gamma$ a process $t$ produces output $\alpha$ or output $\beta$ with expected probability $\widetilde{c}$, and the former (respectively, the latter) output has probability $\widetilde{a}$ (respectively, $\widetilde{b}$), then it produces the latter (respectively, the former) output with probability $\widetilde{c}-\widetilde{a}$ (respectively $\widetilde{c}-\widetilde{b}$). Note that also this rule can be seen as an abbreviation for an instance of ``sampling''  from $n$ premises each of which is the result of introducing $+$ on a single experiment whose output is either $\alpha$ or $\beta$, followed by an instance of ``sampling'' of either output:

\begin{example}
Consider an application of the rule E$+$ used to derive the expected probability of one among different possible outputs of a given process:
   \begin{prooftree}
    \AxiomC{$\vdots$}
    \noLine
    \UnaryInfC{$\Gamma, x_{t}:\alpha_{a},x_{t}:\beta_{b}\vdash t_{n} : (\alpha+\beta)_{\tilde{a+b}}$}
    \AxiomC{$\Gamma, x_{t}:\alpha_{a},x_{t}:\beta_{b}\vdash t_{n}:\beta_{\tilde{b}}$}
    \RightLabel{E$+$}
\BinaryInfC{$\Gamma, x_{t}:\alpha_{a}, x_{t}:\beta_{b}\vdash t_{n}:\alpha_{\widetilde{(a+b)-b}}$}

  \end{prooftree}

followed by a sampling: 
\begin{prooftree}
         \AxiomC{$\Gamma, x_{t}:\alpha_{a},x_{t}:\beta_{b}\vdash t^1:\alpha^1$}
    \AxiomC{$\dots$}
     \AxiomC{$\Gamma, x_{t}:\alpha_{a},x_{t}:\beta_{b}\vdash t^{n}:\alpha^n$}
     \RightLabel{sampling}
     \TrinaryInfC{$\Gamma, x_{t}:\alpha_{a},x_{t}:\beta_{b}\vdash t_{n} : \alpha_{f'}$}
\end{prooftree}

\end{example}
\bigskip

\begin{example}
Consider a fair die $d$ whose distribution of outputs is encoded in $\Gamma$:

\begin{prooftree}
\AxiomC{$\mathcal{D}_{1}$}
\RightLabel{expectation}
\UnaryInfC{$\Gamma \vdash d_{6} : 1_{\widetilde{1/6}}$}
\AxiomC{$\mathcal{D}_{2}$}
\RightLabel{expectation}
\UnaryInfC{$\Gamma \vdash d_{6} : 3_{\widetilde{1/6}} $}
\RightLabel{I$+$}
\BinaryInfC{$\Gamma \vdash d_{6} : (1+3)_{\widetilde{1/3}} $}
\end{prooftree}

now followed by a sampling whose observed outputs are the series $1$, $3$, $5$, $1$, $2$, $6$ in a sample of $6$ throws. In this sample the frequency of $1$ or $3$ outputted by $d$ is given in our syntax as follows

\begin{prooftree}
\AxiomC{$\Gamma\vdash d^{1}:(1+...+6)^{1}$}
\AxiomC{$\dots$}
\AxiomC{$\Gamma\vdash d^{6}:(1+...+6)^{6}$}
\RightLabel{sampling}
\TrinaryInfC{$\Gamma \vdash d_{6} : 1_{3/6}$}
\end{prooftree}


\end{example}

Note that in both examples above, we have two conclusions with the same output, one which expresses the expected probability, one which expresses actual frequency over a number of executed trials. In the following it will be our aim to define operators to measure their distance.

The typing rule I$\times$ says that if two distinct processes (and thus independent: recall that we denote by different terms distinct processes, hence the distinct distributions $\Gamma, \Delta$) $t$ and $u$ produce expected probabilities $\widetilde{a}$ and $\widetilde{b}$ for outputs $\alpha$ and $\beta$ and samples of size $n$ under distributions $\Gamma$ and $\Delta$ respectively; the expected probability of jointly getting output $\alpha$ and $\beta$ from the pair $\langle t,u\rangle$ under the joint distributions $\Gamma,\Delta$, is given by $\widetilde{a\cdot b}$. 
By E$\times$, given a process which provides two distinct output types with associated probability $\tilde{c}$, and knowing the expected probability $\tilde{a}$ of the first composing process in the given sample, we infer the probability of the second process to produce the second output as $\widetilde{c/a}$; the dual rule has $snd$ in the second premise and $fst$ in the conclusion.

\begin{example}\label{ex:conj}
{ 
Let $d$ and $g$ be two dice associated with distributions $\Delta$ and $\Gamma$ respectively. We use distinct names for the distribution of distinct processes, although they might turn out to be identical, e.g., in the case of two fair dice. Assuming independence on $\Delta,\Gamma$, output expectation on a sample of $18$ throws of dice $d$ and $g$ is as follows:

\[\infer[\textrm{I}\times]{\Delta, \Gamma \vdash \langle d, g \rangle_{18} : (1 \times 2)_{\widetilde{1/9}} }{\infer[\textrm{expectation}]{\Delta \vdash d_{18} : 1_{\widetilde{2/3}}}{} &\infer[\textrm{expectation}]{\Gamma \vdash g_{18} : 2_{\widetilde{1/6}} }{}}\]




An observation on a sampling of $18$ throws of dice $d$ might shows the series:\[4,1,2,1,1,1,6,1,3,1,1,1,3,1,1,1,1,5\]
and of $18$ throws of dice $g$ the series:
\[1,3,5,2,3,1,3,4,5,2,6,4,1,3,1,3,4,2\]


The joint sampling, constructed by pairing the $n$th result of the first dice with the $n$th result of the second dice, is the following:
\[(4,1),(1,3),(2,5),\underline{(1,2)},(1,3),(1,1),(6,3),(1,4),(3,5),\]
\[\underline{(1,2)},(1,6),(1,4),(3,1),(1,3),(1,1),(1,3),(1,4),(5,2)\]
The experiment was a success. Indeed, among the $18$ pairs, exactly two are of the form $(1,2)$. Hence, as predicted, the frequency of $(1,2)$ in this list is $2/18=1/9$. This outcome of the experiment on the pair of dice $\langle d,g \rangle $ can de formalized as follows:
\[
\infer{\Delta, \Gamma \vdash \langle d, g \rangle_{18} : (1 \times 2)_{2/18} }{ \Delta, \Gamma \vdash \langle d, g \rangle^{1} : (4 \times 1)^{1} \dots  \Delta, \Gamma \vdash \langle d, g \rangle^{18} : (5 \times 2)^{18}}\]

}



\end{example}
\bigskip

\begin{example}
Given two fair dice $d,g$ whose distributions are encoded in respectively $\Delta, \Gamma$. Again assuming independence, we could have  $$\Delta \times \Gamma = \{ \langle x_d, x_g \rangle : 1\times 1_{1/36}, \langle x_d, x_g \rangle : 1\times 2_{1/36}, \dots, \langle x_d, x_g \rangle : 6\times 6_{1/36} \}$$ 

While $\langle d,g\rangle_{24} : (n\times m)_{\tilde{n}\cdot \tilde{m}}$ expresses the expected probability $\tilde{n}\cdot \tilde{m}$ of any pair $(n\times m)$ over $24$ launches, for any specific such sequence of launches like the following: $(5,5)$, $(2,1)$, $(3,6)$, $(3,1)$, $(1,4)$, $(5,3)$, $(5,4)$, $(3,6)$, $(6,3)$, $(5,2)$, $(2,3)$, $(3,3)$, $(2,3)$, $(4,1)$, $(4,5)$, $(1,2)$, $(6,1)$, $(6,6)$, $(3,6)$, $(6,6)$, $(6,6)$, $(6,5)$, $(4,4)$, $(2,6)$, the frequency of any specific output results from an application of the sampling rule, such as:

\begin{prooftree}
\AxiomC{$\Delta, \Gamma\vdash \langle d,g\rangle^{1} : (5\times 5)^{1}$}
\AxiomC{\dots}
\AxiomC{$\Delta, \Gamma\vdash \langle d,g\rangle^{24} : (2\times 6)^{24}$}
\RightLabel{sampling}
\TrinaryInfC{$\Delta\times \Gamma \vdash \langle d,g\rangle_{24} : (1\times 2)_{1/24} $}
\end{prooftree}

\end{example}

The rule I$\rightarrow$ says that, if assuming a variable has value $\alpha$ with probability $a$ we can execute a process $t$ which, executed $n$ times, shows output $\beta$ with expected probability $\tilde{b}$, then we can construct the term $[x]t_{n}$ which taken in input a process of value $\alpha$ with probability $a$, will return a process of output value $\beta$ with expected probability $\tilde{b}$ depending on $a$. Note that in $[a]\tilde{b}$, the sub-expression $[a]$ is only a notation to keep track of the theoretical probability assigned to the output in the antecedent. The corresponding elimination E$\rightarrow$ allows to verify the expected probability of $\beta$: it considers a term $[x]t_{n}$ of type $(\alpha \rightarrow \beta)$ which has expected probability $\tilde{b}$ depending on the probability $a$ of value $\alpha$, and providing in place of $x$ a process $u$ with expected probability $\tilde{a}$ formulated under a fresh random variable $y:\alpha_{a}$, it returns the process $t$ which will display output $\beta$ with an expected probability $\widetilde{a}\cdot\widetilde{b}$. The frequency of $\beta$ in a sample of $n$ experiments should therefore approximate this value.

\begin{example}
Consider a coin $c$, which we assume to be fair, i.e., whose assumed probability distribution on its outputs has the form

$$
\Gamma:=\{x_{c}:H_{0.5}, x_{c}:T_{0.5}\}
$$

We might devise a system which if the fair coin  $c1$ returns $Head$, it will make another coin $c2$ return $Tail$ with a frequency of $45\%$ of the times over $1000$ throws:

\begin{prooftree}
\AxiomC{$\Gamma, x_{c1}:H_{0.5}\vdash c2_{1000}:T_{\widetilde{0.45}}$}
\RightLabel{I$\rightarrow$}
\UnaryInfC{$\Gamma\vdash [x]c2_{1000}:(H\rightarrow T)_{[0.5]\widetilde{0.45}}$}
\end{prooftree}
{ In order to verify that the system is actually working in this way, we shall throw $c1$ a $1000$ times. If we obtain the information that it returns heads a number of times which is close enough to the $50\%$ of the total number of throws, we can use  this information to obtain information on $c2$. Indeed, on the basis of the  output of $c1$, $c2$ will return tails $45\%$ of the times in which $c1$ returns heads. Therefore, we know that $c2$ will return tails the $45\%$ of the $50\%$ of the $1000$ throws:}
\begin{prooftree}
\AxiomC{$\vdash [x_{c1}]c2:(H\rightarrow T)_{[0.5]\tilde{0.45}}$}
\AxiomC{$x_{c1}:H_{0.5}\vdash c1_{1000}:H_{\widetilde{0.5}}$}
\RightLabel{E$\rightarrow$}
\BinaryInfC{$x_{c1}:H_{0.5}\vdash c2_{1000}.[c1_{1000}:H]: T_{\widetilde{0.225}}$}
\end{prooftree}
%
This final judgment says that we expect tails from the second coin and heads in the first coin with a frequency of $22.5\%$ over $1000$ times.
\end{example}

\subsection{Bayesian Updating}\label{subsec:bayesian}

Random variables and experiments can be linked by the formulation of a form of $\rightarrow$ connective by which random variables are updated using the information generated by trials. The rules in Figure \ref{fig:rulesprior} express such process of updating prior probabilities in the light of new information provided by experiments. It corresponds to the use of the $\rightarrow$ connective introduced from a judgement of the form $x:\alpha^{1}_
{a}, \dots, x:\alpha^{n}_{a^{n}}\vdash y:\beta_{b}$, and eliminated with an additional premise of the form $\vdash t_{n}:\alpha_{f}$. The formula for $b$ implements the Bayesian Updating in the specific case of a finite prior and a binomial likelihood function, which is the case we can express within our calculus. In particular, the prior probability associated with the point $i$ under consideration is $b_i$, and its likelihood is $\binom{n}{n\cdot f} a_i^{f\cdot n} (1-a_i)^{n-f\cdot n}$. The denominator marginalizes as usual over all points in the prior. Note that the binomial coefficient gets simplified out as it appears both in the numerator and denominator of the expression for $b$.

In the introduction, assuming that the probability $b$ of variable $y$ being assigned output $\beta$ resulted from a sampling of $n$ experiment within a model in which probabilities associated to random variables $x, \dots, x^{n}$ to have values $\alpha, \dots, \alpha^{n}$ is respectively $a, \dots, a^{n}$, we introduce a list of implications which take all the exclusive hypotheses assigning a given probability $a_{i}$ to variables $x_{i}$ having value $\alpha_{i}$ and assigning a given variable $y$ probability $b$ to have value $\beta$; in the elimination: given such an implication with a list of exclusive hypotheses, and the result of a sampling on an experiment outputting $\alpha$ with frequency $f$, we update the prior probability $b_{i}$ of $\beta$ to $b$.

\begin{figure}

\begin{prooftree}
\AxiomC{$x:\alpha_{a_1} \vdash y:\beta_{b_1}$, $\dots$, $x:\alpha_{a_m} \vdash y:\beta_{b_m}$ }
\AxiomC{ $\sum_{i=1}^m b_i = 1$ }
\RightLabel{I-P}
\BinaryInfC{ $\{ \vdash [x]y:(\alpha \rightarrow \beta)_{[a_i]b_i} \mid 1\leq i\leq m \}$ }
\end{prooftree}

\begin{prooftree}
\AxiomC{ $\{ \vdash [x]y:(\alpha \rightarrow \beta)_{[a_i]b_i} \mid 1\leq i\leq m,\sum_i b_i = 1\}$ }
\AxiomC{ $\Gamma, x:\alpha_{a_{i}} \vdash t_n : \alpha_{f}$ }
\RightLabel{E-P}
\BinaryInfC{ $\Gamma, x:\alpha_{a_{i}} \vdash (y.(t_n:\alpha_f)) : \beta_{b}$ }
\end{prooftree}
where \[ b = \frac{\overbrace{\binom{n}{n\cdot f} a_i ^{f\cdot n} (1-a_i)^{n-f\cdot n}}^{likelihood}\overbrace{b_i}^{prior}}{\underbrace{\binom{n}{n\cdot f} \sum_{j=1}^m a_j ^ {f\cdot n} (1-a_j)^{n-f\cdot n} b_j}_{marginal}} = \frac{a_i ^{f\cdot n} (1-a_i)^{n-f\cdot n}b_i}{ \sum_{j=1}^m a_j ^ {f\cdot n} (1-a_j)^{n-f\cdot n} b_j} \]
    \caption{Rules to Update Prior Probabilities}
    \label{fig:rulesprior}
\end{figure}

\begin{example}\label{ex:priorupdate}
In this example, we show how we can model the process of updating a discrete prior in a bayesian fashion. Let us assume we have five coins of three types. Coins of type $A$ are fair, whereas coins of types $B$ and $C$ are biased, with a probability of $0.8$ and $0.9$ of landing heads respectively. We have $2$ coins of type $A$, $1$ coin of type $B$ and $2$ coins of type $C$. We blindly draw one of these five coins and toss it three times. Three coin tosses give that the coin lands heads $2$ times. Our prior knowledge can be modelled with the three following sequents:
\[x:\Heads_{0.5} \vdash y:\textit{Draw}_{2/5}\]
\[x:\Heads_{0.8} \vdash y:\textit{Draw}_{1/5}\]
\[x:\Heads_{0.9} \vdash y:\textit{Draw}_{2/5}\]
that we denote by $\mathcal{P}$. They encode the discrete prior in Figure \ref{fig:prior}.
We focus on the first of these three sequents, corresponding to the hypothesis that the coin we have drawn is of type $A$. We now consider data $\mathcal{D}=\{ \Heads, \Heads, \Tails \}$, which we represent with the sequent \[x : \Heads_{0.5} \vdash c_{3}:\Heads_{2/3}\]and calculate the posterior probability of hypothesis $x:\Heads_{0.5} \vdash y:\textit{Draw}_{2/5}$ as follows:
\begin{prooftree}
    \AxiomC{$\mathcal{P}$}
    \AxiomC{$x:\Heads_{0.5} \vdash c_{3}:\Heads_{2/3}$}
    \BinaryInfC{$ x:\Heads_{0.5} \vdash y.(c_3:\Heads_{2/3}) : \textit{Draw}_{a} $}
\end{prooftree}
where \[a = \frac{0.5^2 \cdot (1-0.5) \cdot (2/5)}{0.5^2 \cdot (1-0.5) \cdot (2/5)+ 0.8^2 \cdot (1-0.8) \cdot (1/5) + 0.9^2 \cdot (1-0.9) \cdot (2/5)} \approx 0.46 \]

Therefore, the hypothesis that $x:\Heads_{0.5}$ has increased its probability from $2/5$ to $0.46$.

\end{example}

\begin{figure}
\centering
\begin{tikzpicture}
\begin{axis}[xlabel={Probability of $\Heads$},ylabel={Prior probability}]
\addplot[ybar] plot coordinates
	{(0.5,2/5) (0.8,1/5) (0.9,2/5)};
\end{axis}
\end{tikzpicture}
\caption{A graphical representation of the discrete prior in Example \ref{ex:priorupdate}. The $x$ axis represents the probability of the coin landing $\Heads$; the $y$ axis represents the corresponding prior probability. }\label{fig:prior}
\end{figure}

Note that once the prior has been updated in the conclusion of the previous tree, it is still possible to introduce again the implication 

$$\vdash [x](y.c_{3}:Heads_{2/3}):(Heads \rightarrow \textit{Draw})_{[0.5]a}$$ 
and update it again with some new data.

\subsection{Trust Fragment}\label{subsec:trust}

TPTND is designed to verify trustworthiness of probabilistic computations through the rules in Figure \ref{fig:TrustFragment}. 

{ The intuition we want to implement in the calculus is that the control of a certain distance between a distribution of reference and the observed probability of an event is formally checkable. The choice of a confidence interval with respect to the expected probability of an event is a design choice which can be otherwise devised, see e.g. \cite{DBLP:conf/aiia/PrimieroD22} for a version using KL as a distance in order to evaluate for biases. Moreover, while the rule mimics in a certain sense standard statistical inference properties, the calculus allows to combine such analysis with an automated deduction approach, which also allows to define some useful structural properties on the derivation mimicking the probabilistic process, thereby allowing formal properties to be proven in the next section. Formally, we interpret the above as follows. An experiment $\Delta\vdash u_{n}:\alpha_{f}$ where $\Delta$ is possibly unknown -- that is, $\Delta$ is of the form $\{x:\alpha_{[0,1]}, \dots,  x:\omega_{[0,1]}\}$ as obtained by the ``unknown'' rule in Figure \ref{fig:distributions} -- should be evaluated according to some metric against its intended or expected model,  in turn expressed as a transparent distribution $\Gamma$. If the difference between the frequency $f$ of output $\alpha$ in $n$ experiments $u$ and its theoretical probability $a$ in the relevant intended and transparently formulated distribution $\Gamma$ remains below a critical threshold, parametric with respect to the number of samples, then the process $u$ outputting $\alpha$ with frequency $f$ can be considered trustworthy (rule IT). The parametric threshold $\epsilon(n)$ is domain-specific and it may depend on the application. Then one can start reasoning on processes considered trustworthy: if $u_{n}:\alpha_{f}$ is trustworthy under distribution $\Gamma$, then there must exist an interval of probability values $[a-\epsilon(n),a+\epsilon(n)]$  for output $\alpha\in \Gamma$ within which $u_{n}:\alpha_{f}$ can be correctly sampled from $\Gamma$ (rule ET). Note that Contraction from Section \ref{subsec:structural} can be applied to select the most appropriate value within such range.}

Dual rules for negative trust can be devised. Rule IUT for the introduction of negative trust (or Untrust) expresses the failure on the condition for trust: if the value set for the difference between the frequency $f$ of the output $\alpha$ in the experiment and the theoretical probability $a$ in the relevant intended and transparently formulated distribution $\Gamma$ is surpassed, the process can be labeled untrustworthy. The rule EUT serves the purpose of eliminating the negative trust function, i.e., it allows to start reasoning on processes considered untrustworthy: if $u_{n}:\alpha_{f}$ is considered untrustworthy under distribution $\Gamma$, then there exists a family of probability values in the interval $[0,1]$ excluding the interval $[a-\epsilon(n),a+\epsilon(n)]$  for output $\alpha\in \Gamma$ within which $u_{n}:\alpha_{f}$ can be correctly sampled from $\Gamma$. Hence, the rule identifies the range of theoretical values which make $f$ a fair frequency, and again Contraction from Section \ref{subsec:structural} can be applied to select the most appropriate value within such range.

\begin{figure*}
    \begin{prooftree}
            \AxiomC{$ \Gamma \vdash x:\alpha_a$}
            \AxiomC{$ \Delta\vdash u_n : \alpha_{f} $}
            \AxiomC{$ \mid a - f \mid \leq \epsilon(n) $}
            \RightLabel{IT}
            \TrinaryInfC{$ \Gamma,\Delta \vdash \Trust(u_n : \alpha_{f} ) $}
        \DisplayProof
        \\[1em]

         \AxiomC{$\Gamma \vdash \Trust(u_n : \alpha_{f}) $}
            \RightLabel{ET}
            \UnaryInfC{$\Gamma , x_{u} : \alpha_{[a-\epsilon(n),a+\epsilon(n)]} \vdash u_{n} : \alpha_{f} $}
\end{prooftree}

  \begin{prooftree}
            \AxiomC{$ \Gamma \vdash x:\alpha_a$}
            \AxiomC{$\Delta\vdash u_n : \alpha_{f} $}
            \AxiomC{$ \mid a - f\mid > \epsilon(n) $}
            \RightLabel{IUT}
            \TrinaryInfC{$ \Gamma, \Delta \vdash  \UTrust(u_n : \alpha_{f} ) $}
            \DisplayProof
        \\[1em]

            \AxiomC{$\Gamma \vdash \UTrust(u_n : \alpha_{f} ) $}
            \RightLabel{EUT}
            \UnaryInfC{$\Gamma,x_{u} : \alpha_{[0,1]-[a-\epsilon(n),a+\epsilon(n)]} \vdash u_{n} : \alpha_{f} $}
            \end{prooftree}
     
    \caption{Trust Fragment}
    \label{fig:TrustFragment}
\end{figure*}

Note, in particular, that the first and second premise of each introduction rule allow for distinct distributions. While it is required that the first context expresses a distribution in which it is included the theoretical probability for a program of output $\alpha$, this means that it is possible to evaluate for trustworthiness against it a program $u$ of the same type even if the latter is executed under a distinct, and possibly opaque or even unknown distribution $\Delta$. The easy case is obviously where $\Gamma \equiv \Delta$, but otherwise $\Delta$ might be inaccessible (i.e. formally unknown) while $\Gamma$ represents the intended or assumed model of $u$. 

Trust evaluation, as mentioned, is parametric in the number of executions, and this renders it a potentially dynamic process. For instance, one may take $\epsilon(n)$ to correspond to the $95\%$ confidence interval under the Normal approximation to the Binomial Distribution. Then it is possible to derive the trustworthiness of some untrustworthy process (and vice versa) with some (very small) probability. The probability of this happening can be made arbitrarily small by extending the sample associated with the process.

\begin{example}\label{ex:OpaqueContext} Let $u$ be a die which we have been provided with and we know nothing about, its distribution being denoted as $\Delta$; and let $\Gamma$ be the distribution defined in Example \ref{ex:fairdie}, representing a fair die. In this example, we focus on output $5$. Let's assume that experimenting under the unknown context $\Delta$ gives us the following result:

\begin{prooftree}
        \AxiomC{$\dots$}
        \AxiomC{$\dots$}
        \AxiomC{$\dots$}

        \RightLabel{sampling}
        \TrinaryInfC{$\Delta \vdash u_{10} : 5_{1/2}$}
\end{prooftree}

Our aim is to decide whether the system $u$ can be assumed to represent a sample under a fair distribution over outputs $\{ 1, 2, \dots, 6 \}$, i.e., whether $u$ can be therefore considered safe to use during a betting game. To this aim, we use the Introduction of Negative Trust, where we use the (exact) Binomial $95\%$ confidence interval to set our threshold:

\begin{prooftree}
        \AxiomC{$\Gamma \vdash x:5_{1/6}$}
        \AxiomC{$\Delta\vdash u_{10} : 5_{1/2}$}
        \AxiomC{$1/6 \notin [0.19, 0.81]$}
        \RightLabel{IUT}
        \TrinaryInfC{$\Gamma, \Delta \vdash UTrust(u_{10}:5_{1/2})$}
\end{prooftree}
that is, we do not trust process $u$ to be extracted from a distribution which sufficiently approximates $\Gamma$.

Obviously, things may change as new evidence is collected. Let suppose we are provided with some more experimental evidence about process $u$, we then can use rule ``update'':
        \begin{scprooftree}{0.90}
                \AxiomC{$\Delta \vdash u_{10} : 5_{1/2}$}
                \AxiomC{$\Delta \vdash u_{20} : 5_{3/20}$}
                \RightLabel{upd}
                \BinaryInfC{$\Delta\vdash u_{30} : 5_{4/15}$}

                \AxiomC{$\Gamma\vdash x:5_{1/6}$}
                %
                \AxiomC{$1/6 \in [0.12, 0.46]$}
                \RightLabel{IT}
                \TrinaryInfC{$\Gamma, \Delta \vdash \Trust(u_{30}:5_{4/15} ) $}
        \end{scprooftree}
i.e., based on a total of $30$ runs of process $u$, now we trust it to be extracted from the transparent and fair distribution $\Gamma$. 
\end{example}


\begin{example}
Suppose we are given a commercial, closed-source software to automatically shortlist CVs according to a set of criteria. To this aim, we consider the output of the classification algorithm to fall into one of the following four categories: (1) male, shortlisted, (2) male, not shortlisted, (3) female, shortlisted, (4) female, not shortlisted. To verify that the underlying algorithm does not have an inherent gender bias, we test whether we trust the frequency of any output of interest over these four classes to reflect the actual gender distribution in the current Italian population, as encoded in a distribution $\Gamma$. A proof that this is indeed the case would produce a certificate in the form of a natural deduction proof, that would go as follows:

\begin{scprooftree}{0.90}
    \AxiomC{$\{\} :: \distribution$}
    \RightLabel{extend}
    \UnaryInfC{$\{ x : \textit{female}_{0.52} \} :: \distribution$}
    \AxiomC{$\{\} \vdash c_{10}:\textit{female}_{0.3}$}
    \AxiomC{$ 0.52 \in [0.0667,0.6525] $}
    \RightLabel{IT}
    \TrinaryInfC{$\{ x : \textit{female}_{0.52} \}  \vdash \Trust(c_{10} : \textit{female}_{0.3})$}
\end{scprooftree}
\bigskip
\noindent where we assumed a $95\%$ confidence interval. Note that a shortlist selection which would provide strictly less than $2/10$ or more than $9/10$ of female candidates would be considered untrustworthy according to this model. Note that if we sampled $30$ CVs instead, shortlisting $7$ female CVs, we would derive untrustworthiness:

\begin{scprooftree}{0.90}
    \AxiomC{$\{\} :: \distribution$}
    \RightLabel{extend}
    \UnaryInfC{$\{ x : \textit{female}_{0.52} \} :: \distribution$}
    \AxiomC{$\{\} \vdash c_{30}:\textit{female}_{7/30}$}
    \AxiomC{$ 0.52 \notin [0.1,0.42] $}
    \RightLabel{IUT}
    \TrinaryInfC{$\{ x : \textit{female}_{0.52} \}  \vdash UTrust(c_{30} : \textit{female}_{7/10} )$}
\end{scprooftree}
We would in fact consider untrustworthy any number of shortlisted CVs by female candidates strictly smaller than $10$ or bigger than $21$. Assume now a different population of interest:
{\small
\begin{scprooftree}{0.90}
    \AxiomC{$\{\} :: \distribution$}
    \RightLabel{extend}
    \UnaryInfC{$\{ x : \textit{female}_{0.7} \} :: \distribution$}
    \AxiomC{$\{\} \vdash c_{10}:\textit{female}_{3/10}$}
    \AxiomC{$ 0.7 \notin [0.0667,0.6525] $}
    \RightLabel{IUT}
    \TrinaryInfC{$\{ x : \textit{female}_{0.7} \}  \vdash UTrust(c_{10} : \textit{female}_{3/10} )$}
\end{scprooftree}
}
\end{example}

\subsection{Structural Rules}\label{subsec:structural}
    
Structural rules in a proof-theoretic setting are used to determine valid operations on contexts. Under our interpretation, such rules are needed to establish coherent manipulation of assumptions on probability distributions under which either other random variables are assigned a value, or processes (programs, experiments) are sampled. In particular, our structural rules are essential to two aims:

\begin{enumerate}
    \item define the admissible extensions of value assignments to random variables through the ``extend'' type definition;
\item define the extraction of plausible values of relevant random variables from experiments performed under multiple hypotheses and, more generally, under unknown distributions.
\end{enumerate}
 We provide constrained versions of the standard rules holding for our system in Figure \ref{fig:struct}.

\begin{figure}
 \begin{prooftree}
 	\AxiomC{$\Gamma 
  \vdash 
 	z:\beta_{b}$}
  	\AxiomC{$\Delta\vdash y:\gamma_{g}$}
  	\AxiomC{$\Gamma \indep \Delta$}
   	\RightLabel{Weakening}
  	\TrinaryInfC{$\Gamma, \Delta
   \vdash z:\beta_{b}$}
  \end{prooftree}

 
 \begin{prooftree}
	\AxiomC{$\Gamma,x_{t}:\alpha_{a_{1}},\dots, x_{t}:\alpha_{a_{n}}\vdash t_{n}:\beta_{f}$}
   	\RightLabel{Contraction}
  	\UnaryInfC{$\Gamma,x_{t}:\alpha_{fun(a_{1},\dots, a_{n})} \vdash t_{n}:\beta_{f}$}
 \end{prooftree}
     \caption{Structural Rules}
         \label{fig:struct}
     \end{figure}

The Weakening rule expresses the following principle: given the probability of a random variable $z$ to have value $\beta$ is $b$, provided the value assigned to some variable $x$ is $\alpha$; if the latter assignment is independent of a distribution $\Delta$, then the starting conditional probability can be extended by additional assumptions $\Delta$ and the probability $b$ is left unaltered. In other words: a given probability distribution can be extended at will without inferring new probabilities, as long as no new dependencies are introduced. To show this case, and illustrate how Weakening by a confounding variable modifies the inferred output, we illustrate a number of examples.

\begin{example}[Confounding]\label{ex:confounding}
Consider the following graph:

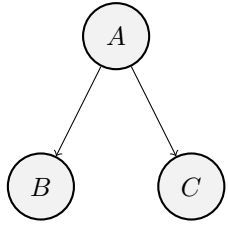
\begin{center}
    \begin{tikzpicture}
        \node[state] at (1,2) (a) {$A$};
        \node[state] (b) {$B$};
        \node[state, right of=b] (c) {$C$};
        \draw (a) edge (b);
        \draw (a) edge (c);
    \end{tikzpicture}
\end{center}

This graphical model means that $P(A,B,C) = P(B \mid A) P(C\mid A) P(A)$, i.e., $A$ is a \textit{confounding} variable as $B$ and $C$ depend on $A$ but do not directly influence each other. In this example, $A$, $B$ and $C$ are probabilistic routines with outputs $0$ and $1$. $A$ is a ``master'' process in that it can decide to force ``slave'' processess $B$ and $C$ to output specific values. $A$ has a $50\%$ chance of generating output $1$ and $50\%$ chance of producing output $0$; if $A$ outputs $0$, it forces $B$ and $C$ to output $0$; if $A$ outputs $1$, $B$ and $C$ both have $50\%$ probability of producing $1$ and $50\%$ probability of producing $0$. In this case, $A$ acts as a \textit{confounder}, in the sense that one may observe a dependency between $B$ and $C$ that, however, is not due to an explicit communication between the two processes.

In TPTND this false dependency is discarded, as the dependency from $A$ is clearly stated, while $B,C$ are shown to be coordinated processes by logical conjunction under the same assumption. The entire graph is expressed analytically by the following trees (where the shared context trivially satisfies our definition of $\indep$ since we can derive $\{ x_{A} : 0,  x_{A} : 1\} :: \distribution$):
\begin{scprooftree}{0.90}
    \AxiomC{$ x_{A} : 0\vdash y_{B}:0_{\tilde{1}}$}
    \AxiomC{$ x_{A} : 0\vdash z_{C}:0_{\tilde{1}}$}
    \RightLabel{I$\times$}
\BinaryInfC{$ x_{A} : 0\vdash \langle y_{B},z_{C}\rangle:(0\times 0)_{\tilde{1}}$}
    \RightLabel{I$\rightarrow$}
\UnaryInfC{$ \vdash [x_{A}]\langle y_{B},z_{C}\rangle:(0\rightarrow (0\times 0))_{\tilde{1}}$}
\end{scprooftree}

\begin{scprooftree}{0.90}
\AxiomC{$ x_{A} : 1\vdash y_{B}:1_{\tilde{0.5}}$}
\AxiomC{$ x_{A} : 1\vdash z_C:1_{\tilde{0.5}}$}
  \RightLabel{I$\times$}
\BinaryInfC{$ x_{A} : 1\vdash \langle y_{B},z_{C}\rangle:(1\times 1)_{\tilde{0.25}}$}
    \RightLabel{I$\rightarrow$}
\UnaryInfC{$ \vdash [x_{A}]\langle y_{B},z_{C}\rangle:(1\rightarrow (1\times 1))_{\tilde{0.25}}$}
\end{scprooftree}

\end{example}
\bigskip

\begin{example}[Chain]
Consider the following different graph:

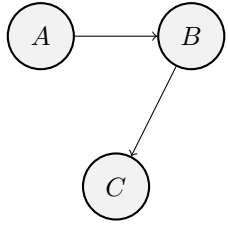
\begin{center}
    \begin{tikzpicture}
        \node[state] at (1,2) (a) {$B$};
        \node[state, left of=a] (b) {$A$};
        \node[state] (c) {$C$};
        \draw (b) edge (a);
        \draw (a) edge (c);
    \end{tikzpicture}
\end{center}

This graphical model means that $P(A,B,C) = P(C \mid B ) P(B\mid A) P(A)$, i.e., $C$ is a variable depending on $B$, whose effect in turn depends on $A$. In this example, $A$, $B$ and $C$ are probabilistic routines with outputs $0$ and $1$. $A$ is a ``master'' process in that it can decide to force ``slave'' process $B$ and this in turn process $C$ to output specific values. $A$ has a $50\%$ chance of generating output $1$ and $50\%$ chance of producing output $0$; if $A$ outputs $0$, it forces $B$ to output $0$; under the $50\%$ probability that $B$ outputs $0$, $C$ has $50\%$ probability of producing $1$ and $50\%$ probability of producing $0$. 

{\small 
\begin{scprooftree}{0.90}

 \AxiomC{$y_{B}.[x_{A} : 0]:0_{a}\vdash z_{C}:1_{\tilde{0.5}}$}
\RightLabel{I$\rightarrow$}
 \UnaryInfC{$ \vdash [y_{B}.[x_{A} : 0]]z_{C}:(0\rightarrow 1)_{[a]\tilde{0.5}}$}
 \AxiomC{$ x_{A} : 0\vdash y_{B}:0_{\tilde{1}}$}
    \RightLabel{I$\rightarrow$}
\UnaryInfC{$ \vdash [x_{A}]y_{B}:(0\rightarrow 0)_{\tilde{1}}$}
\AxiomC{$x_{A} : 0\vdash x_{A} : 0_{\tilde{a}}$}
    \RightLabel{E$\rightarrow$}
\BinaryInfC{$x_{A} : 0\vdash y_{B}.[x_{A} : 0]:0_{\tilde{a}\cdot \tilde{1}}$}
 \RightLabel{Weakening}
 \BinaryInfC{$ x_{A} : 0\vdash [y_{B}.[x_{A} : 0]]z_{C}:(0\rightarrow 1)_{[a]\tilde{0.5}}$}
\end{scprooftree}
}
Note that an evaluation of the probability of the variable $z_{C}$ having value $1$ depends ultimately on the value which will be assigned to variable $a$ in the E$\rightarrow$ by the premise $x_{A}:0\vdash x_{A}:0_{\tilde{a}}$. Note moreover how the condition of $C$ is expressed by a probability $P(B\mid A)$ and this corresponds in the tree above to the fact that the context of $z_{C}$ is of the form $y_{B}.[x_{A} : 0]:0_{a}$ with probability $a$.

\end{example}

The previous example is illustrative of a variable $y_{B}$ which is not independent of $x_{A}$, hence the latter could not be used to weaken the context of $z_{C}$. On the other hand, the following example introduces the situation of a variable $y_{C}$ dependent on $x_{A}$, and for which one could also state the dependency from a distinct variable $x_{B}$:

\begin{example}[Collider]
Consider now the following example:

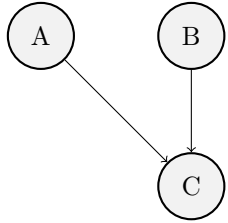
\begin{center}
\begin{tikzpicture}
 \node[state] (a) {A};
 \node[state, right of=a] (b) {B};
 \node[state, below of=b] (c) {C};

 \draw (a) edge (c); 
 \draw (b) edge (c);
\end{tikzpicture}
\end{center}

In this example, the output of variable $C$ is a collider of $A$ and $B$. Again, $A$, $B$ and $C$ are probabilistic routines with outputs $0$ and $1$. $A$ is a ``master'' process in that it can decide to force ``slave'' process $C$: $A$ has a $50\%$ chance of generating output $1$ and $50\%$ chance of producing output $0$:
 \begin{prooftree}
   \AxiomC{$x_{A}:0_{0.5}, x_{A}:1_{0.5} \vdash x_{A}:0_{0.5}$}
     \AxiomC{$x_{A}:0_{0.5}, x_{A}:1_{0.5}\vdash x_{A}:1_{0.5}$}
    \RightLabel{I$+$}
    \BinaryInfC{$x_{A}:0_{0.5}, x_{A}:1_{0.5}\vdash x_{A}:(0+1)_{0.5+0.5}$}
  \end{prooftree}if $A$ outputs $0$, it forces $C$ to output $0$, while if $A$ outputs $1$, $C$ will have only $50\%$ chance to output $1$:
{
{\tiny\begin{scprooftree}{0.80}
 \AxiomC{$x_{A}:0_{0.5} \vdash y_{C}:0_{\widetilde{1}}$}
\RightLabel{I$\rightarrow$}
\UnaryInfC{$\vdash [x_{A}]y_{C}:0 \rightarrow 0_{[0.5]\widetilde{1}}$}
\AxiomC{$x_{A}:0_{0.5} \vdash x_{A}:0_{0.5}$}
\RightLabel{E$\rightarrow$}
\BinaryInfC{$x_{A}:0_{0.5}\vdash [x_{A}]y_{C}:0_{\widetilde{0.5}}$}
\AxiomC{$x_{A}:0_{0.5}\vdash y_{C}:1_{\widetilde{0}}$}
\RightLabel{I$\rightarrow$}
\UnaryInfC{$\vdash [x_{A}]y_{C}:(1\rightarrow 1 ) _{[0.5]\widetilde{0}}$}
\AxiomC{$x_{A}:0_{0.5} \vdash x_{A}:0_{0.5}$}
\RightLabel{E$\rightarrow$}
\BinaryInfC{$x_{A}:0_{0.5}\vdash [x_{A}:0_{0.5}]y_{C}:1_{\widetilde{0}}$}
\RightLabel{I$+$}
\BinaryInfC{$x_{A}:0_{0.5}\vdash [x_{A}:0_{0.5}]y_{C}:(0+1)_{\widetilde{0.5}}$}
\end{scprooftree}}
}
{
{\tiny\begin{scprooftree}{0.80}
\AxiomC{$x_{A}:1_{0.5} \vdash y_{C}:1_{\widetilde{0.5}}$}
\RightLabel{I$\rightarrow$}
\UnaryInfC{$\vdash [x_{A}]y_{C}:(1\rightarrow 1)_{[0.5]\widetilde{0.5}}$}
\AxiomC{$x_{A}:1_{0.5} \vdash x_{A}:1_{0.5}$}
\RightLabel{E$\rightarrow$}
\BinaryInfC{$x_{A}:1_{0.5}\vdash y_{C}.[x_{A}:1_{0.5}]:1_{\widetilde{0.25}}$}
\AxiomC{$x_{A}:1_{0.5} \vdash y_{C}:0_{\widetilde{0.5}}$}
\RightLabel{I$\rightarrow$}
\UnaryInfC{$
\vdash [x_{A}]y_{C}:(1\rightarrow 0)_{[0.5]\widetilde{0.5}}$}
\AxiomC{$x_{A}:1_{0.5} \vdash x_{A}:1_{0.5}$}
\RightLabel{E$\rightarrow$}
\BinaryInfC{$x_{A}:1_{0.5}\vdash y_{C}.[x_{A}:1_{0.5}]:0_{\widetilde{0.25}}$}
\RightLabel{I$+$}
\BinaryInfC{$x_{A}:1_{0.5}\vdash y_{C}.[x_{A}:1]:(0+1)_{\widetilde{0.5}}$}
\end{scprooftree}}
}
\bigskip


Note that here the total probability of $C$ outputting $(0+1)$ is $1$, but this is expressed under different conditions, namely of process $A$ outputting either $0$ or $1$. While we cannot derive this, by inspection of the two derivation trees, we can select the probability of $\vdash y_{C}.[x_{A}:0]:0_{0.5}$ in the leftmost branch of the first tree and the probability of $\vdash y_{C}.[x_{A}:1]:0_{0.25}$ and compute the overall probability of $C$ having output $0$ when $A$ has value $(0+1)$ is $(0.5+0.25)=0.75$. Similarly, the probabilities of $\vdash y_{C}.[x_{A}:0]:1_{0}$ and $\vdash y_{C}.[x_{A}:1]:1_{0.25}$ gives us the overall probability of output $1$ from $C$ as $(0+0.25)=0.25$. Note, however, that this computation is not derivable.

Now consider that $B$ is also a ``master'' process in that it can decide to force ``slave'' process $C$: $B$ has a $50\%$ chance of generating output $1$ and $50\%$ chance of producing output $0$:
\begin{prooftree}
   \AxiomC{$x_{B}:0_{0.5} \vdash x_{B}:0_{0.5}$}
   \AxiomC{$x_{B}:1_{0.5}\vdash x_{B}:1_{0.5}$}
   \RightLabel{I$+$} 
   \BinaryInfC{$x_{B}:0_{0.5}, x_{B}:1_{0.5}\vdash x_{B}:(0+1)_{0.5+0.5}$}
  \end{prooftree}

if $B$ outputs $1$, it forces $C$ to output $0$, while if $B$ outputs $0$, $C$ will have only $50\%$ chance to output $1$:
{
{\tiny\begin{scprooftree}{0.80}
 \AxiomC{$x_{B}:1_{0.5} \vdash y_{C}:0_{\widetilde{1}}$}
\RightLabel{I$\rightarrow$}
\UnaryInfC{$\vdash [x_{B}:1_{0.5}]y_{C}:(1\rightarrow 0)_{[0.5]\widetilde{1}}$}
\AxiomC{$x_{B}:1_{0.5} \vdash x_{B}:1_{0.5}$}
\RightLabel{E$\rightarrow$}
\BinaryInfC{$x_{B}:1_{0.5}\vdash y_{C}.[x_{B}:1_{0.5}]:0_{\widetilde{0.5}}$}
\AxiomC{$x_{B}:1_{0.5} \vdash y_{C}:1_{\widetilde{0}}$}
\RightLabel{I$\rightarrow$}
\UnaryInfC{$\vdash [x_{B}:1_{0.5}]y_{C}:(1\rightarrow 1)_{[0.5]\widetilde{0}}$}
\AxiomC{$x_{B}:1_{0.5} \vdash x_{B}:1_{0.5}$}
\RightLabel{E$\rightarrow$}
\BinaryInfC{$x_{B}:1_{0.5}\vdash y_{C}.[x_{B}:1_{0.5}]:1_{\widetilde{0}}$}
\RightLabel{I$+$}
\BinaryInfC{$x_{B}:1_{0.5}\vdash y_{C}.[x_{B}:1_{0.5}]:(0+1)_{\widetilde{0.5}}$}
\end{scprooftree}}
}

{
{\tiny\begin{scprooftree}{0.80}
 \AxiomC{$x_{B}:0_{0.5} \vdash y_{C}:1_{\widetilde{0.5}}$}
\RightLabel{I$\rightarrow$}
\UnaryInfC{$\vdash [x_{B}:0_{0.5}]y_{C}:(0\rightarrow 1)_{[0.5]\widetilde{0.5}}$}
\AxiomC{$x_{B}:0_{0.5} \vdash x_{B}:0_{0.5}$}
\RightLabel{E$\rightarrow$}
\BinaryInfC{$x_{B}:0_{0.5}\vdash y_{C}.[x_{B}:0_{0.5}]:1_{\widetilde{0.25}}$}
\AxiomC{$x_{B}:0_{0.5} \vdash y_{C}:0_{\widetilde{0.5}}$}
\RightLabel{I$\rightarrow$}
\UnaryInfC{$\vdash [x_{B}:0_{0.5}]y_{C}:(0\rightarrow 0)_{[0.5]\widetilde{0.5}}$}
\AxiomC{$x_{B}:0_{0.5} \vdash x_{B}:0_{0.5}$}
\RightLabel{E$\rightarrow$}
\BinaryInfC{$x_{B}:0_{0.5}\vdash y_{C}.[x_{B}:0_{0.5}]:0_{\widetilde{0.25}}$}
\RightLabel{I$+$}
\BinaryInfC{$x_{B}:0_{0.5}\vdash y_{C}.[x_{B}:0_{0.5}]:(0+1)_{\widetilde{0.5}}$}
\end{scprooftree}}
}
\bigskip


Note that here the total probability of $C$ outputting $(0+1)$ is again $1$, but this is expressed under different conditions, namely of process $B$ outputting either $0$ or $1$. By inspection of the two derivation trees, we can select the probability of $\vdash y_{C}.[x_{B}:1]:0_{0.5}$ in the leftmost branch of the first tree and the probability of $\vdash y_{C}.[x_{B}:0]:0_{0.25}$ and compute the overall probability of output $0$ as $(0.5+0.25)=0.75$. Similarly, the probabilities of $\vdash y_{C}.[x_{B}:1]:1_{0}$ and $\vdash y_{C}.[x_{B}:0]:1_{0.25}$ gives us the overall probability of output $1$ from $C$ dependent from $B$ is $(0+0.25)=0.25$. Note, however, that this computation is not derivable.

In a similar vein, to compute the overall probability that $y_{C}:0$ given $x_{A}:1$ and $x_{B}:1$, we can inspect the rightmost branch of the second of the previous two trees and the leftmost branch of the first of the two trees above: $(0.25+0.5)-(0.25\cdot0.5)=0.625$. Or the overall probability that $y_{C}:0$ given $x_{A}:0$ and $x_{B}:0$, inspecting the rightmost branch of the second of the previous two trees and the leftmost branch of the first of the two trees above: $(0.5+0.25)-(0.5\cdot0.25)=0.625$. Again, this computation is available by inspection of the derivation trees, but not directly derivable.

On the other hand, if we formulate $y_{C}$ under the  condition $x_{A}:(0+1)$ (respectively $x_{B}:(0+1)$), we can express the direct dependency of $C$ having output $(0+1)$ deriving it from the two disjuncts: 
\begin{scprooftree}{0.80}
\AxiomC{$x_{A}:(0+1)_{1} \vdash y_{C}:0_{0.75}$}
 \AxiomC{$x_{A}:(0+1)_{1} \vdash y_{C}:1_{0.25}$}
\RightLabel{I$+$}
\BinaryInfC{$x_{A}:(0+1)_{1}\vdash y_{C}:(0+1)_{1}$}
\end{scprooftree}

\begin{scprooftree}{0.80}
     \AxiomC{$x_{B}:(0+1)_{1} \vdash y_{C}:0_{0.75}$}
 \AxiomC{$x_{B}:(0+1)_{1} \vdash y_{C}:1_{0.25}$}
\RightLabel{I$+$}
\BinaryInfC{$x_{B}:(0+1)_{1}\vdash y_{C}:(0+1)_{1}$}
\end{scprooftree}

\begin{scprooftree}{0.80}
    \AxiomC{$x_{A}:(0+1)_{1}, x_{B}:(0+1)_{1}\vdash y_{C}:0_{0.75}$}
    \AxiomC{$x_{A}:(0+1)_{1}, x_{B}:(0+1)_{1}\vdash y_{C}:1_{0.25}$}
\RightLabel{I$+$}
\BinaryInfC{$x_{A}:(0+1)_{1}, x_{B}:(0+1)_{1}\vdash y_{C}:(0+1)_{1}$}
\end{scprooftree}

\bigskip

\end{example}

The Contraction rule expresses a form of plausibility test on hypotheses. Consider a probability distribution $\Gamma$ which includes more than one distinct hypothesis on the theoretical probability $a$ of some process with output value $\alpha$, on which the evaluation of an experiment $t$ with output $\beta$ with frequency $f$ depends. The general case will be that $\Gamma$  is unknown, so that for any output variable $\alpha\in \Gamma$, the entire set of possible probability values $[0,1]$ is assigned. A contraction on such distribution is a function $fun[0,1]$ which extracts one value $x_{t}:\alpha_{a}$. The function $fun$ can be chosen at will; a sensible example is the Maximum Likelihood function $ML(a,a')$ which returns the value making the frequency $f$ the most plausible as follows:

\begin{prooftree}
	\AxiomC{$\Gamma,x_{t}:\alpha_{a},x_{t}:\alpha_{a'}\vdash t_{n}:\alpha_{f}$}
   	\RightLabel{Contraction}
  	\UnaryInfC{$\Gamma,x_{t}:\alpha_{ML(a,a')} \vdash t_{n}:\alpha_{f}$}
 \end{prooftree}
where $ML(a,a') = \text{argmax}_{x\in\{a,a'\}} ( x^{\frac{f}{n}}(1-x)^{1-\frac{f}{n}} )$.

\begin{example}[Contraction] Let $\Delta$ be the unknown distribution associated to a die $d$ with possible outputs $\{1, \dots, 6\}$, i.e.:
\begin{prooftree}
\AxiomC{$1::\Output$, $\dots$, $6::\Output$}
\RightLabel{unknown}
\UnaryInfC{$\Delta :: \distribution$}
\end{prooftree}
Technically, $\Delta$ could be re-written as the distribution which associates to each output $\{1, \dots, 6\}$ the value $[0,1]$. Consider now $10$ throws of the die with sequence of outputs $1$, $6$, $6$, $2$, $6$, $3$, $4$, $6$, $1$, $1$. Then,
\begin{prooftree}
    \AxiomC{}
    \RightLabel{sampling}
    \UnaryInfC{$\Delta \vdash d_{10} : 6_{4/10} $}
\end{prooftree}

Applying contractions multiple times can be abbreviated as follows:
\begin{prooftree}
    \AxiomC{$\Delta \vdash d_{10} : 6_{4/10} $}
    \RightLabel{Contraction}
    \UnaryInfC{$\Delta\setminus \{ x:6_a \mid a\notin[0,0.4)\cup(0.4,1] \} \vdash d_{10} : 6_{4/10} $}
\end{prooftree}
\bigskip
where the context in the conclusion denotes the set obtained by the unknown distribution in the premise, removing from it the declaration which assigns any probability different than $0.4$ to $x$ having output $6$, i.e., $\Delta\setminus \{ x:6_a \mid a\in[0,0.4)\cup(0.4,1] \} = \{ x:1_{[0,1]}, x:2_{[0,1]}, x:3_{[0,1]}, x:4_{[0,1]}, x:5_{[0,1]}, x:6_{0.4} \}$.

\end{example}
\bigskip

\begin{example}
Consider the distribution $\Gamma$ representing a fair coin, i.e., with $50\%$ probability of landing Heads when tossed. Assume $\Delta$ expresses a coin which we know nothing about, i.e $\Delta$ is unknown, which when tossed 10 times lands Heads only once. Under the confidence range expressed by $[0.0025,0.4450]$, the coin is untrustworthy.

\begin{prooftree}
        \AxiomC{$\Gamma, x:H_{0.5}::distribution$}
        \AxiomC{$\Delta\vdash c_{10} : H_{1/10}$}
        \AxiomC{$0.5 \notin [0.0025,0.4450]$}
        \RightLabel{IUT}
        \TrinaryInfC{$\Gamma, x:H_{0.5}, \Delta \vdash \UTrust(c_{10}:H_{1/10})$}
    
\end{prooftree}

We can now decide to assume the coin $c$ to be biased with a probability of only $0.3$ of landing Heads when tossed: this can be done by applying an appropriate instance of contraction. Under this new assumption, the result of 10 tosses of the coin landing Heads only once is trustworthy.

\begin{prooftree}
\AxiomC{$x_{c}:H_{[0,1]}, \dots,  x_{c}:H_{[0,1]}\vdash c_{10}:H_{1/10}$}
\RightLabel{contraction}
\UnaryInfC{$x_{c}:H_{0.3}\vdash c_{10}:H_{1/10}$}
\AxiomC{$0.3 \in [0.0025,0.4450]$}
        \RightLabel{IT}
        \BinaryInfC{$x:H_{0.3}\vdash \Trust(c_{10}:H_{1/10})$}
\end{prooftree}    
\end{example}
\bigskip

Finally, the Cut rule is admissible as a detour of I$\rightarrow$ and E$\rightarrow$ rules when both the major and the minor premises result from sampling:

 \begin{prooftree}
          \AxiomC{$ \Delta, x : \alpha_{a} \vdash t_n : \beta_{\widetilde{b}} $}
          \RightLabel{I$\rightarrow$}
          \UnaryInfC{$ \Delta \vdash [x]t_n : (\alpha \rightarrow \beta)_{[a]\widetilde{b}} $}
      \AxiomC{$\Gamma \vdash u_n:\alpha_{\widetilde{a}} $}
          \RightLabel{E$\rightarrow$}
          \BinaryInfC{$ \Gamma,\Delta \vdash t_{n}.[u:\alpha]:\beta_{\widetilde{a\cdot b}} $}
     \end{prooftree}

\section{Metatheory}\label{sec:meta}

Our meta-theoretical analysis aims at showing properties of computational terms with respect to expected behavior, as well as considering inference trees of TPTND to establish how logical rules behave with respect to inferring trustworthiness. Hence, a judgement $\Gamma\vdash t_{n}:\alpha_{f}$ is considered in the light of the computational transformations formalized by the relation $\mapsto_{p}$ for an expression of the form $t_{n}:\alpha_{f}$.


Our first step is to reduce expected probabilities when they depend on deterministic outputs. Deterministic substitution has the following meaning: consider a process $t$ which has output $\beta$ with expected probability $\tilde{b}$, under the assumption of a variable having value of type $\alpha$ with probability $a$; whenever output $\alpha$ is deterministically obtained by a process $u$, process $t$ has output $\beta$ with probability $\tilde{b}$ under the condition that $u$ is executed and $\alpha$ is obtained.

\begin{theorem}[Deterministic Substitution]\label{theorem:determsub}
If $x\!:\!\alpha_{a} \vdash t:\beta_{\tilde{b}}$ and $x_{u}:\alpha_{a} \vdash u\!:\!\alpha$, 
{
then there exists a term
$s=t.\mathcal{C}_1[u].\, \dots\,.\mathcal{C}_n[u]$ where $1\leq n $ and each $\mathcal{C}_i[u]$ is obtained by zero or more logical eliminations on $u$, and $\vdash s :\beta_{\tilde{b}}$.} 

\end{theorem}

\begin{proof}
By structural induction on the antecedent of the first premise. To show the reduction to output $\beta_{\tilde{b}}$, for each case of $x\!:\!\alpha_{a}$ in the induction we first build one or more dependent terms of the form $[x]t:(\alpha ' \rightarrow \beta)_{[a]\tilde{b}}$, then consider the deterministic output $u\!:\!\alpha$ in the second derivation and use it to infer the application
$t.\mathcal{C}_1[u].\, \dots\,.\mathcal{C}_n[u]$:



\begin{itemize}
    \item For $\alpha$ atomic, the case is immediate;
    
    \item For $\alpha\equiv \gamma\times\delta$, there are terms $\langle v,z\rangle$ obtained from independent distributions such that the following tree is a reduction to the base case:
  \begin{scprooftree}{0.75}
  \AxiomC{$x_v\!:\!\gamma , y_z\!:\!\delta \vdash t:\beta_{\tilde{b}}$}
\RightLabel{I$\rightarrow$}
\UnaryInfC{$ y_z\!:\!\delta \vdash [x]t:(\gamma \rightarrow \beta)_{[1]\tilde{b}}$}
          \AxiomC{$x_v\!:\!\gamma , y_z\!:\!\delta\vdash \langle v,z\rangle:(\gamma\times\delta)$}
          \RightLabel{E$\times$}
          \UnaryInfC{$x_v\!:\!\gamma , y_z\!:\!\delta \vdash fst(\langle v,z\rangle):\gamma$}
          \RightLabel{E$\rightarrow$}
          \BinaryInfC{$ x_v\!:\!\gamma, y_{z}\!:\!\delta\vdash t.[fst(\langle v,z\rangle):\gamma]:\beta_{\tilde{b}}$}
         \RightLabel{I$\rightarrow$}
         \UnaryInfC{$ x_v\!:\!\gamma , \vdash [y]t.[fst(\langle v,z\rangle):\gamma]:(\delta\rightarrow \beta)_{[1]\tilde{b}}$}
         
         \AxiomC{$x_{v}\!:\!\gamma, y_{z}\!:\!\delta\vdash \langle v,z\rangle:(\gamma\times\delta)$}
         \RightLabel{E$\times$}
         \UnaryInfC{$x_{v}\!:\!\gamma, y_{z}\!:\!\delta \vdash snd(\langle v,z\rangle):\delta$}
         \RightLabel{E$\rightarrow$}
         \BinaryInfC{$x_v\!:\!\gamma, y_z\!:\!\delta\vdash t.[fst(\langle v,z\rangle):\gamma].[snd(\langle v,z\rangle):\delta]: \beta_{1\cdot\tilde{b}}$}       
  \end{scprooftree}

    
     \item For $\alpha\equiv \gamma + \delta$ is a metavariable for the distribution of possible exclusive outputs of a given term $u$. Then either $u:\gamma$ or $u:\delta$ (and respectively the other term has probability $=0$); then use
     an additional assumption of type $\mathtt{O}^{\bot}$ to select either one through $E+$, and the base case applies. 

    \item For $\alpha\equiv \gamma\rightarrow \gamma$ (as an instance of the rule $update$ with unary executions and deterministic outputs), there are terms $u\equiv [y]v$ and $v$ such that the following tree is a reduction to the base case: 
\[
\infer[\mathrm{E}\rightarrow]{
y_v:\gamma\vdash t.[ v.[v:\gamma]:\gamma]:\beta_{1\cdot\tilde{b}}
}{
\infer[\mathrm{I}\rightarrow]{\vdash [y]t:( \gamma\rightarrow \beta)_{\tilde{b}}}{y_{v}:\gamma\vdash t:\beta_{\tilde{b}}}
&
\infer[\mathrm{E}\rightarrow]{y_v:\gamma\vdash v.[v:\gamma]:\gamma}{
\infer[\mathrm{I}\rightarrow]{\vdash [y]v:\gamma \rightarrow \gamma }{ 
y_{v}:\gamma\vdash v:\gamma 
} & y_{v}:\gamma\vdash v:\gamma }}\]
\end{itemize}
\end{proof}

We consider the computational semantics of samplings  to prove properties about frequencies and expected probabilities for such terms.  
Recall that a computational term $t$ is said to probabilistically reduce to a term $t'$ if $t$ can be transformed into $t'$ according to one of the rewriting rules from Figures \ref{fig:eval},\ref{fig:eval2}.  First, we show that for atomic terms obtained by applications of reduction rules in Figure \ref{fig:eval}, the value of the frequency of any given output approximates, in the limit, its expected probability.

\begin{theorem}
\label{thm:atomic-convergence} 
For any atomic term $t$ such that $\Gamma , x_{t}:\alpha_{a} \vdash t_n:\alpha _{\tilde{a}}$, and sequence of reductions 
\[\mathcal{L}_1, t_{n_0}:\alpha _{f_0}\mapsto \dots \mapsto \mathcal{L}_m, t_{n_m}:\alpha _{f_m}\mapsto \dots \]with $n_0 <\dots <n_m<\dots $ that only contains event, sampling and update rule applications, let us define $\nu$ as the function such that $\nu_t(m)=f_m$. We have then that
\[
\lim _{m\rightarrow \infty} \vert  a-\nu_t( m)\vert=0
\]with probability $1$.
\end{theorem}
\begin{proof}
   {Let us define a random variable $X$ with possible values $0$ and $1$ and expected value $E(X)=a$. A sequence of random variables $X_1, \dots , X_n$ then encodes the count of how many times an output of type $\alpha$ is obtained during an experiment consisting of $n$ executions of the term $t$. The expected value $E(X)=a$ indeed exactly corresponds to the probability that an application of the $\textrm{event}^{\mapsto t}$ rule produces a typed term of the form $t:\alpha$ indicating that $t$ yielded an output of type $\alpha$. This is due to the fact that, since $x_{t}:\alpha_{a} $ is in our context and our context is meant to formalize the behaviour of the process $t$, the $\textrm{event}^{\mapsto t}$ rule produces a typed term $t:\alpha $ with probability $a$. 

Notice that a sequence $X_1, \ldots , X_n$ such that $\frac{X_1, \ldots , X_n}{n}=E(X)=a$ encodes a count of the number of outcomes of type $\alpha$ during an experiment consisting of $n$ executions of $t$ such that exactly $a$ of the total outcomes was of type $\alpha$. 

In case all random variables $X_1, \ldots , X_n$ have the same expected value $E(X)$, the strong law of large numbers states that 
\[\lim _{n \rightarrow \infty} \left(\frac{X_1 + \ldots + X_n}{n}   \right)=E(X) \]with probability $1$. This, hence, holds for any sequence $X_1, \ldots , X_n$ of occurrences of the random variable $X$ defined above.

Since \begin{itemize}
\item each sequence $X_1, \ldots , X_n$ such that $\frac{X_1, \ldots , X_n}{n}=E(X)=a$ corresponds to an experiment consisting of $n$ executions of $t$ such that exactly $a$ of the total outcomes was of type $\alpha$, 
\item the closer $\frac{X_1, \ldots , X_n}{n}$ gets to $E(X)=a$, the closer the frequency of outcomes of type $\alpha$ produced by the $\textrm{event}^{\mapsto t}$ rules applications gets to the expected probability $a$, 
\item the value $f_m$ of typed atomic terms of the form $t_m:\alpha_{f_m}$ produced by term evaluation rules directly and exclusively depends on the frequency with which a term $t:\alpha$ is obtained by the $\textrm{event}^{\mapsto t}$ rule; 
\end{itemize}the strong law of large numbers implies that, with probability $1$, for $m$ approaching infinity, any reduction \[\mathcal{L}_1, t_{n_0}:\alpha _{f_0}\mapsto \dots \mapsto \mathcal{L}_m, t_{n_m}:\alpha _{f_m} \]
with $n_0 <\dots <n_m $ is such that $f_m$ is closer and closer to $a$. Therefore, by definition of $\nu$, we have
\[
\lim _{m\rightarrow \infty} \vert  a-\nu_t( m)\vert=0
\]with probability $1$.
}
\end{proof}

Before closing this result under logical operations, let us introduce some standard notation.

\begin{definition} By $\mapsto^*$, as per standard practice, we denote the reflexive and transitive closure of the $\mapsto $ relation.
\end{definition}

We can now close Theorem \ref{thm:atomic-convergence} under those logical operations which are related to the the reduction rules in Figure \ref{fig:eval2}.

\begin{theorem}\label{thm:complex-convergence}
For any term $t$ such that $\Gamma \vdash t_{n_0}:\alpha _{f_0}$ and $\Gamma \vdash t_{n_0}:\alpha _{\tilde{a}}$  
and sequence of reductions 
\[\mapsto^*\mathcal{L}_1, t_{n_0}:\alpha _{f_0}\mapsto^* \dots \mapsto^* \mathcal{L}_m, t_{n_m}:\alpha _{f_m}\mapsto^* \dots \]starting from the empty list and with $n_0 <\dots <n_m<\dots $, let us define $\nu$ as the function such that $\nu_t(m)=f_m$ (where, if $f_m$ contains sub-expressions $[a]$ related to $\rightarrow$ types, we ignore them). We have then that
\[
\lim _{m\rightarrow \infty} \vert  a-\nu_t( m)\vert=0
\]with probability $1$.
\end{theorem}
\begin{proof}
Since the subscript of $t$ increases at the displayed steps of the reduction sequence
\[\mathcal{L}_1, t_{n_0}:\alpha _{f_0}\mapsto^* \dots \mapsto^* \mathcal{L}_m, t_{n_m}:\alpha _{f_m}\mapsto^* \dots \]
we have that the atomic terms occurring in $t_{n_0}, \dots ,t_{n_m}, \dots$ occur with increasing subscripts themselves when they occur as typed terms during our reduction. Notice that the increase in the subscript of $t$ cannot be exclusively due to I$^{\mapsto}\times$ rule applications since the type of $t$ is always the same at the displayed reduction steps (even though it might change at the non-displayed steps, during which, potentially, logical operation rules of elimination and introduction are applied). Since the atomic terms occurring in $t_{n_0}, \dots ,t_{n_m}, \dots$ have increasing subscripts themselves when they occur as typed terms in our reduction, Theorem \ref{thm:atomic-convergence} guarantees that the frequencies expressed by these atomic terms when typed tends, for the subscripts of the terms increasing to infinity, to have difference $0$ with respect to the expected probability of the types of these atomic terms. By induction on the number of term evaluation rule applications,
we show that if the frequency expressed by the atomic terms occurring in $t_{n_0}, \dots ,t_{n_m}, \dots$  when they occur as typed terms tends to have difference $0$ with respect to the expected probability of their types, then the frequencies $f_0 , \dots, f_m \dots$ expressed by the typed terms $t_{n_0}:\alpha _{f_0}, \dots , t_{n_m}:\alpha _{f_m} \dots$ tends to have difference $0$ with respect to the expected probability $a$ of $t$ yielding $\alpha$ as output.

 Let us reason by induction on the number of term evaluation rules applied.

If only one evaluation rule has been applied, it must be an $\textrm{event}^{\mapsto t}$ rule and $\alpha$ is atomic. In this case, 
our hypothesis on the atomic terms occurring in $t_{n_0}, \dots ,t_{n_m}, \dots$ directly gives us what we need to prove.

Suppose then that for any typed term occurrence obtained by less than $r$ term evaluation rule applications  corresponding to logical operations the statement holds, we prove that it holds also for any typed term occurrence obtained by $r$ term evaluation rule applications  corresponding to logical operations.

We reason on the $r\,$th term evaluation rule applied to obtain the considered term. 
\begin{itemize}

\item If the term has been obtained by I$^{\mapsto}+$, then $t_m:( \beta^0 + \beta^1)_f $ and, by inductive hypothesis, the statement holds for both occurrences $t_m:\beta^0_{g_0}$ and $t_m:\beta^1_{g_1}$, from which $t_m:( \beta^0 + \beta^1)_f $ has been produced by I$^{\mapsto}+$,  with respect to the expected probabilities $b_0$ and $b_1$ such that 
$\Gamma \vdash t_m:\beta^0 _{\tilde{b_0}}$ and  $\Gamma \vdash t_m:\beta^1_{\tilde{b_1}}$. 
Since  
$\Gamma \vdash t_m:(\beta^0+\beta^1) _{\tilde{b_0+b_1}}$
and $t_m:( \beta^0 + \beta^1)_{f}$ for $f= g_0+g_1$, we have that the statement holds also for $ t_m:( \beta^0 + \beta^1)_{f}$.
 
\item  If the term has been obtained by I$^{\mapsto}\times$, then there must be two subterms $t^0,t^1$ of $t_m:  ({    \beta^0 \times \beta ^1})_f$ such that $t^0_p:\beta^0_{g_0}, t^1_p:\beta ^1_{g_1}$ and $t_m:\alpha_f=\langle t_0,t_1\rangle_{p}:({    \beta^0 \times \beta ^1})_{g_0\cdot g_1}$. By inductive hypothesis, the statement holds for both occurrences $t^0_p:\beta^0_{g_0}$ and $t^1_p:\beta ^1_{g_1}$ with respect to the expected probabilities $b_0$ and $b_1$ such that $\Gamma \vdash t^0:\beta^0 _{\tilde{b_0}}$ and  $\Gamma \vdash t^1:\beta^1_{\tilde{b_1}}$. Since  
$\Gamma \vdash \langle  t^0 ,t^1\rangle_{p}:(\beta^0\times\beta^1) _{\tilde{b_0\cdot b_1}}$
and $t_m:\alpha_f=\langle t^0,t^1\rangle_{p}:({    \beta^0 \times \beta ^1})_{g_0\cdot g_1}$, we have that the statement holds also for $t_m:\alpha_f$.

\item If the term has been obtained by I$^{\mapsto}\rightarrow$, then the term is  $t_m:  ({    \beta^0 \rightarrow \beta ^1})_f$ and there must be a subterm $t^1$ of it such that $t^1_m:\beta^1_{g_1}$ and $t_m:\alpha_f=[x_u]t^1:
  ({    \beta^0 \rightarrow \beta ^1})_{[b_0]g_1}$. By inductive hypothesis, the statement holds for  the occurrence $t^1_m:\beta^1_{g_1}$ with respect to the expected probability $b_1$ such that 
$\Gamma \vdash t^1_m :\beta^1 _{\tilde{b_1}}$
Since 
$\Gamma \vdash t^1_m:(\beta^0\rightarrow\beta^1) _{[b_0]\tilde{b_1}}$ 
and $t_m:\alpha_f=[x_u]t^1:
(\beta^0 \rightarrow \beta^1)_{[b_0]g_1}$, 
we have that the statement holds also for $t_m:\alpha_f$. Notice that the sub-expression $[b_0]$ is ignored for the purpose of computing the difference between $[b_0]g_1$ and $[b_0]b_1$.
  
\item If the term has been obtained by E$^{\mapsto}+_L$, then the term is $t_m : \alpha_{f_1-f_2} $. By inductive hypothesis,
 the statement holds for the term $t_m : (\alpha+\beta)_{f_1} $, from which $t_m$ has been obtained by E$^{\mapsto}+_L$, with respect to the expected probability $b_1$ such that 
$\Gamma \vdash t_m :(\alpha + \beta)_{\tilde{b_1}}$
 and the statement holds for the term $t_m : \beta_{f_2}$ (the premise of the E$^{\mapsto}+_L$ application used to obtain the term, which can be obtained with two less term evaluation rule applications than $t_m : \alpha_{f_1-f_2}$) with respect to the expected probability $b_2$ such that 
$\Gamma \vdash t_m  : \beta _{\tilde{b_2}}$. 
 Since $\Gamma \vdash t : \alpha _{\tilde{b_1-b_2}}$, we have that the statement holds also for $t_m : \alpha_{f_1-f_2} $.

\item The case of E$^{\mapsto}+_R$ is analogous.

\item If the term has been obtained by E$^{\mapsto}\times_L$, then the term is $\textit{fst}(t)_{m} : \alpha_{f_1/f_2} $. By inductive hypothesis,
 the statement holds for the  term $t_{m} : (\alpha\times \beta) _{f_1} $,  from which $\textit{fst}(t)_{m}$ is  obtained by E$^{\mapsto}\times_L$, with respect to the expected probability $b_1$ such that 
$\Gamma \vdash \langle \textit{fst}(t),u \rangle_{m} :(\alpha \times \beta)_{\tilde{b_1}}$ 
 and the statement holds for the term  $\snd(u)_{m} : \beta_{f_2}$  (the premise of the E$^{\mapsto}\times_L$ application used to obtain the term, which can be obtained with two less term evaluation rule applications than $\textit{fst}(t)_{m} : \alpha_{f_1/f_2} $) with respect to the expected probability $b_2$ such that $\Gamma \vdash u_{m}: \beta _{\tilde{b_2}}$. Since $\Gamma \vdash \textit{fst}(t)_{m}: \alpha _{\tilde{b_1/b_2}}$, we have that the statement holds also for $\textit{fst}(t)_{m} : \alpha_{f_1/f_2} $.
 

\item The case of E$^{\mapsto}\times_R$ is analogous.

\item If the term has been obtained by E$^{\mapsto}\rightarrow$, then the term is $t_m . [u_m : \beta] : \alpha_{f_2 \cdot f_1} $. By inductive hypothesis, the statement holds for the terms from which $t_m . [u_m : \beta] $ has been obtained by E$^{\mapsto}\rightarrow$, that is, $u_m :\beta _{f_2}$ with respect to the expected probability $b_2$ such that $\Gamma \vdash x_u : \beta _{\tilde{b_2}}$ and the term $[x_u]t_{m_1} : (\beta \rightarrow \alpha) _{[b_2]f_1} $  with respect to  $\Gamma \vdash [x_u]t _m  : (\beta \rightarrow \alpha)_{[b_2]\tilde{b_1}}$. Since $\Gamma \vdash t_m .[x_u: \beta]:\alpha _{\tilde{b_2\cdot b_1}}$, we have that the statement holds also for $t_m . [u_m : \beta] : \alpha_{f_2 \cdot f_1}$.

\end{itemize}

\end{proof}

{
The following theorem says that if we construct a term by making explicit the dependencies from the theoretical probabilities of another term, then the distance between the frequency of outputs and their theoretical probability does not increase. On the other hand, if we construct complex terms  or we gather new data on a term, the distance of the output frequencies from the corresponding theoretical probabilities might vary. This distinction is reflected on the formal properties of our evaluation rules. In particular, the rule I$^{\mapsto}\rightarrow$, which makes a term functionally dependent from assumptions on theoretical probabilities, does not influence the distance between output frequency and theoretical probability of outputs. Other rules, instead, do influence this distance. Consider, for instance, the rule 

\[\infer[\mathrm{I}^{\mapsto}+]{\listerm, t_n:\alpha_f ,t_n :\beta_g \quad \mapsto_{1}\quad \listerm,  t_n :(\alpha +\beta )_{f+g}}{}\]

Even if the frequency $f$ is very close to the theoretical probability of $t$ to return $\alpha$, and $g$ is very close to the theoretical probability of $t$ to return the output $\beta$, there is no guarantee that $f+g$ is at least as close to the sum of the theoretical probabilities that $t$ returns an output $\alpha$ and that $t$ returns an output $\beta$. While properties of distance variations of this kind are studied in Theorems \ref{thm:complex-convergence} and \ref{theorem:progress}; in the following theorem, we formally prove that the rule I$^{\mapsto}\rightarrow$ does not influence this distance. This also means that we can always make all dependencies of a term explicit---that is, make it a function of all our theoretical assumptions---without worrying about this distance.

\begin{theorem}[Output Preservation under Explicit Assumptions]\label{th:preservation}
Consider any number $\varepsilon$ and term evaluation reduction $\mathcal{L}, t^1{}',\dots,  t^r{}' \mapsto \mathcal{L},u'$ resulting from an application of I$^{\mapsto}\rightarrow$.
If, for any $i\in\{1, \dots , r, r+1\}$, the following hold:
\begin{itemize}  
\item $t^{r+1}$ is the premise of the rule applied,
\item $t^i{}'=t^i_{n^i}:\beta^i_{f^i}$, 
\item the judgements $\Gamma \vdash t^i_{n^i}:\alpha ^i_{\tilde{a^i}}$ and $\Gamma \vdash t^i_{n^i}:\alpha ^i_{f^i}$ are derivable, and 
\item $\vert f^i-a^i\vert <\varepsilon$;
\end{itemize}
then the following hold as well:
\begin{itemize} 
\item $u'=u_m:\beta _g$, 
\item $\Gamma \vdash u_m:\beta_{\tilde{b}}$ and 
\item $\vert g-b\vert<\varepsilon$.
\end{itemize}
\end{theorem}

\begin{proof}
Suppose that $u'=u_m:\beta _g= [y_s]t^1_{n^1}:(\gamma\rightarrow \alpha^1 )_{[c]f^1}$, for some term $s$, type $\gamma$ and number $c$, has been obtained by I$^{\mapsto}\rightarrow$. Then, by the hypotheses on the term to which we have applied the rule, we have $\vert f^1-a^1\vert<\varepsilon$. Since $f^1=g$ and  $a^1=b$, we have that $\vert g-b\vert <\varepsilon$, as desired. Notice that the sub-expression $[c]$ is ignored for the purpose of computing the difference between $[c]f^1$ and $[c]a^1$.
\end{proof}

}

\begin{theorem}[Subject reduction]
\label{thm:subject-reduction}
If $\mathcal{L} \mapsto \mathcal{L}'$ by an application of a term evaluation rule $r$ with premises belonging to $\Pi$ and, for any $t_{n}:\alpha_f\in\mathcal{L} \cup \Pi$, it holds that $\Gamma \vdash t_{n}:\alpha_f$; then, for any $u_{m}:\beta _g \in\mathcal{L}'$, it holds that $\Gamma \vdash u_{m}:\beta _g$
\end{theorem}
\begin{proof}
Now, any term evaluation rule application that does not concern any introduction or elimination of occurrences of $+,\times$ or $\rightarrow$ can be easily simulated by one application of the sampling rule in Figure \ref{fig:connectiverules}. Let us then consider the case in which a term evaluation rule corresponding to a logical operation is applied. Let us denote by $R^{\mapsto}$ the rule applied. Notice first that any term evaluation rule corresponding to a logical operation is completely analogous to one rule for the introduction or elimination of $+,\times$ or $\rightarrow$ in Figure \ref{fig:singleExperimentRules}. Hence, let us denote by $R^\vdash$ the rule in Figure \ref{fig:singleExperimentRules} corresponding to $R{^\mapsto}$.

Clearly, it is impossible to directly  simulate the term evaluation step obtained by $R^{\mapsto}$ by an application of $R^{\vdash}$ since this rule, being it a single experiment  deduction rule, cannot be applied to judgements of  the form $\Gamma \vdash t_{n}:\alpha_f$ where $f$ is a frequency. Nevertheless,  notice that any derivation $\delta$ with conclusion $\Gamma \vdash t_{n}:\alpha_f$ must consist of upper parts only containing applications of single experiment deduction rules in Figure \ref{fig:singleExperimentRules} and then a lower part only containing sampling and update rules (Figure \ref{fig:connectiverules}). Indeed, the rules for sampling and update are necessary for deriving a conclusion of the form $\Gamma \vdash t_{n}:\alpha_f$, but it is impossible to apply single experiment rules to the conclusions of applications of the rules for sampling and update. Let us denote by $\delta _1 , \dots , \delta _n  $ all largest sub-derivations of $\delta$ that do not contain any sampling or update rule applications and notice that the conclusion of $\delta $ will contain the same types and terms as the conclusions of $\delta _1 , \dots , \delta _n  $,  since sampling and update rules never change any type or the structure of any term. 

Now, in order to derive $ \Gamma \vdash u_{m}:\beta _g$, for any $u_{m}:\beta _g \in\mathcal{L}'$, it is enough to apply $R^{\vdash}$ to the conclusions of some of the sub-derivations $\delta _1 , \dots , \delta _n  $  of the derivation $\delta $ of $\Gamma \vdash t_{n}:\alpha_f$ to obtain a derivation of $\Gamma \vdash u_{m}:\beta _b$. Afterwards, we can apply the sampling rule to suitably chosen premises in order to derive  $\Gamma \vdash u_{m}:\beta _g$ with the desired frequency $g$.
\end{proof}

{
Theorem \ref{th:preservation} implies that making explicit the dependencies of a term on theoretical probability assumptions also preserves trustworthiness. The following corollary formally shows this.

\begin{corollary}\label{lemma:progresstotrust}
 For any term evaluation reduction $\mathcal{L}, t^1{}',\dots,  t^r{}' \mapsto \mathcal{L},u'$ resulting from an application of I$^{\mapsto}\rightarrow$, 
 if, for any $i\in\{1, \dots , r, r+1\}$ where $t^{r+1}$ is the 
 premise of the rule applied, $t^i{}'=t^i_{n^i}:\beta^i_{f^i}$, if $\Gamma \vdash Trust(t^i_{n^i}:\alpha ^i_{f^i})$ is derivable, 
 then $\Gamma \vdash Trust(u_{m}:\beta_{g})$ where $u'=u_m:\beta _g$.
 \end{corollary}
 \begin{proof}
 Suppose that $\mathcal{L}, t^1{}',\dots,  t^r{}' \mapsto \mathcal{L},u'$, where all terms involved comply with the hypotheses of our statement, and that $\Gamma \vdash Trust(t^i_{n^i}:\alpha ^i_{f^i})$ is derivable. Then, by the definition of the IT rule  in Figure \ref{fig:TrustFragment}, we know that there is a number $\epsilon(n^i)$ such that 
 \begin{itemize}
 \item $\Gamma \vdash x^i: \alpha^i _{a^i}$, 
 \item $\Gamma \vdash t^i_{n^i}:\alpha ^i_{f^i} $ and
 \item $\vert a^i-f^i\vert  \leq \epsilon(n^i)$. 
 \end{itemize}But then, by Theorem \ref{th:preservation}, we have that, for $u'=u_m:\beta _g$, both $\Gamma \vdash u_m:\beta_{\tilde{b}}$ and $\vert g-b\vert< \epsilon(m)$ hold.
 Moreover, by Theorem \ref{thm:subject-reduction}, we have that $\Gamma \vdash u_m:\beta _g$. Hence, by the definition of the IT rule in Figure \ref{fig:TrustFragment}, we can  conclude that $\Gamma \vdash Trust(u_{m}:\beta_{g})$ where $u'=u_m:\beta _g$.
 \end{proof}
}

Finally: if a long-enough sequence of reductions occur, then a process will eventually be identified as trustworthy; else, the process remains untrustworthy. We make this explicit by the following final result:


\begin{theorem}[Progress]\label{theorem:progress}
If $\Gamma, x:\beta_a\vdash t_{n}:\beta_{f}$ and $\varepsilon =\epsilon (n)>0$, then either $t_{n+m}:\beta_{f'}$ such that $t_{n}:\beta_{f} \rightarrow^{*} t_{n+m}:\beta_{f'}$ for $m\geq 0$ exists, and $\Gamma\vdash Trust(t_{n+m}:\beta_{f'})$, or $t_{n}:\beta_{f}$ is untrustworthy.
\end{theorem}
\begin{proof}

 {
If $\Gamma \vdash t_n :\beta _{\tilde{b}}$ and 
$\vert a-b\vert \leq \varepsilon =\epsilon (n)>0$ (where $\epsilon (n) $ is the parametric threshold introduced in Section \ref{subsec:trust}),
then Theorems \ref{thm:complex-convergence} and \ref{thm:subject-reduction}, along with the definition of the IT rule in Figure \ref{fig:TrustFragment} -- since $\vert a-b\vert \leq \varepsilon$ -- guarantee that there exists $m\geq 0$
such that $t_{n}:\beta_{f} \rightarrow^{*} t_{n+m}:\beta_{f'}$ and $\Gamma\vdash Trust(t_{n+m}:\beta_{f'})$.

If, otherwise, $\Gamma \vdash t_n :\beta _{\tilde{b}}$ and $\vert a-b\vert > \varepsilon =\epsilon (n)>0 $}, then Theorems \ref{thm:complex-convergence} and \ref{thm:subject-reduction}, along with the definition of the IUT rule in Figure \ref{fig:TrustFragment} -- since $ \vert a-b\vert > \varepsilon $ --  guarantee that there exists $m\geq 0$ such that $t_{n}:\beta_{f} \rightarrow^{*} t_{n+m}:\beta_{f'}$ and $\Gamma\vdash UTrust(t_{n+m}:\beta_{f'})$, which means that the term $u$ is untrustworthy.

\end{proof}

\section{Conclusion and future work}\label{sec:conclusions}

We introduced TPTND, a typed natural deduction system whose judgements express the derivability of probabilistic values for random variables, terms for processes decorated by a sample size and types for output values. The intended interpretation of such judgements is to assert the validity of the probabilistic output expected for a given process under a probability distribution. 

The main use of such a calculus is the evaluation of trustworthy probabilistic computations: trustworthiness is here intended as a property induced by an acceptable distance between output frequency and its theoretical probability, parametric with respect to the sample size under observation. Such notion of trustworthiness can be used to evaluate opaque distributions against intended or desirable models, e.g., in the context of Machine Learning systems. Conditions for trustworthiness are made explicit in TPTND by a safety result, defining the required structure of a TPTND derivation by rules which determine trust in the probability of a given output.

Several extensions of this work are planned, or have already been made. First, the language of TPTND can be extended with properties of the probability distributions in order to introduce bias, and the corresponding metric on trustworthiness can be adapted accordingly as in \citep{DBLP:conf/aiia/PrimieroD22}. Moreover, 
imprecise probabilities can be added to model uncertainty further in the model under observation. Also, a subtyping relation can be formulated to identify classification partitions. Second, the computational semantics for the typed terms of TPTND proposed in \citep{230200958} can represent the basis for the development of a Coq verification protocol for probabilistic trustworthy computations, extending the existing protocol for trust presented in \citep{DBLP:journals/wias/PrimieroB18} with one of the available Coq libraries for probabilistic reasoning, e.g., \url{https://github.com/jtassarotti/polaris}. A variation of TPTND can be devised for modelling processes with finite resources for experiments: in this format, safety can rely on a normal form for terms which is reached after a fixed number of possible experiments. Finally, a relational semantics for this system has been presented in \cite{pk24}.

\section*{Acknowledgments}
All authors thankfully acknowledge the support of the Italian Ministry of University and Research through the projects PRIN2017 n. 20173YP4N3 and PRIN2020 n. 2020SSKZ7R (BRIO - Bias, Risk and Opacity in AI). Fabio Aurelio D'Asaro acknowledges the funding and support of PON ``Ricerca e Innovazione'' 2014-2020 (PON R\&I FSE-REACT EU), Azione IV.6 ``Contratti di ricerca su tematiche Green'' in the context of the project titled ``Il miglioramento dell'algorithmic fairness in ambito assicurativo: Un'analisi concettuale a partire da uno studio di caso''. Francesco Genco and Giuseppe Primiero further acknowledge the support of the Italian Ministry of University and Research through the Project ``Departments of Excellence 2023-2027'' awarded to the Department of Philosophy ``Piero Martinetti''.

\bibliography{samplebib}

\begin{thebibliography}{43}
\providecommand{\natexlab}[1]{#1}
\providecommand{\url}[1]{\texttt{#1}}
\expandafter\ifx\csname urlstyle\endcsname\relax
  \providecommand{\doi}[1]{doi: #1}\else
  \providecommand{\doi}{doi: \begingroup \urlstyle{rm}\Url}\fi

\bibitem[Adams and Jacobs(2015)]{DBLP:conf/types/Adams015}
Robin Adams and Bart Jacobs.
\newblock A type theory for probabilistic and bayesian reasoning.
\newblock In Tarmo Uustalu, editor, \emph{21st International Conference on
  Types for Proofs and Programs, {TYPES} 2015, May 18-21, 2015, Tallinn,
  Estonia}, volume~69 of \emph{LIPIcs}, pages 1:1--1:34. Schloss Dagstuhl -
  Leibniz-Zentrum f{\"{u}}r Informatik, 2015.
\newblock \doi{10.4230/LIPIcs.TYPES.2015.1}.
\newblock URL \url{https://doi.org/10.4230/LIPIcs.TYPES.2015.1}.

\bibitem[Albarghouthi et~al.(2017)Albarghouthi, D'Antoni, Drews, and
  Nori]{3133904}
Aws Albarghouthi, Loris D'Antoni, Samuel Drews, and Aditya~V. Nori.
\newblock Fairsquare: Probabilistic verification of program fairness.
\newblock \emph{Proc. ACM Program. Lang.}, 1\penalty0 (OOPSLA), oct 2017.
\newblock \doi{10.1145/3133904}.
\newblock URL \url{https://doi.org/10.1145/3133904}.

\bibitem[Aldini(2018)]{DBLP:journals/tomacs/Aldini18}
Alessandro Aldini.
\newblock Design and verification of trusted collective adaptive systems.
\newblock \emph{{ACM} Trans. Model. Comput. Simul.}, 28\penalty0 (2):\penalty0
  9:1--9:27, 2018.
\newblock \doi{10.1145/3155337}.
\newblock URL \url{https://doi.org/10.1145/3155337}.

\bibitem[Aldini and Tagliaferri(2019)]{DBLP:conf/esorics/AldiniT19}
Alessandro Aldini and Mirko Tagliaferri.
\newblock Logics to reason formally about trust computation and manipulation.
\newblock In Andrea Saracino and Paolo Mori, editors, \emph{Emerging
  Technologies for Authorization and Authentication - Second International
  Workshop, {ETAA} 2019, Luxembourg City, Luxembourg, September 27, 2019,
  Proceedings}, volume 11967 of \emph{Lecture Notes in Computer Science}, pages
  1--15. Springer, 2019.
\newblock ISBN 978-3-030-39748-7.
\newblock \doi{10.1007/978-3-030-39749-4\_1}.
\newblock URL \url{https://doi.org/10.1007/978-3-030-39749-4\_1}.

\bibitem[Alur et~al.(1996)Alur, Henzinger, and Ho]{489079}
Rajeev Alur, Thomas~A. Henzinger, and Pei-Hsin Ho.
\newblock Automatic symbolic verification of embedded systems.
\newblock \emph{IEEE Trans. Softw. Eng.}, 22\penalty0 (3):\penalty0 181--201,
  mar 1996.
\newblock ISSN 0098-5589.
\newblock \doi{10.1109/32.489079}.
\newblock URL \url{https://doi.org/10.1109/32.489079}.

\bibitem[Bacci et~al.(2018)Bacci, Furber, Kozen, Mardare, Panangaden, and
  Scott]{DBLP:conf/lics/BacciFKMPS18}
Giorgio Bacci, Robert Furber, Dexter Kozen, Radu Mardare, Prakash Panangaden,
  and Dana~S. Scott.
\newblock Boolean-valued semantics for the stochastic {\(\lambda\)}-calculus.
\newblock In Anuj Dawar and Erich Gr{\"{a}}del, editors, \emph{Proceedings of
  the 33rd Annual {ACM/IEEE} Symposium on Logic in Computer Science, {LICS}
  2018, Oxford, UK, July 09-12, 2018}, pages 669--678. {ACM}, 2018.
\newblock \doi{10.1145/3209108.3209175}.
\newblock URL \url{https://doi.org/10.1145/3209108.3209175}.

\bibitem[Boender et~al.(2015)Boender, Primiero, and
  Raimondi]{DBLP:conf/pst/BoenderPR15}
Jaap Boender, Giuseppe Primiero, and Franco Raimondi.
\newblock Minimizing transitive trust threats in software management systems.
\newblock In Ali~A. Ghorbani, Vicen{\c{c}} Torra, H{\"{u}}seyin Hisil, Ali
  Miri, Ahmet Koltuksuz, Jie Zhang, Murat Sensoy, Joaqu{\'{\i}}n
  Garc{\'{\i}}a{-}Alfaro, and Ibrahim Zincir, editors, \emph{13th Annual
  Conference on Privacy, Security and Trust, {PST} 2015, Izmir, Turkey, July
  21-23, 2015}, pages 191--198. {IEEE} Computer Society, 2015.
\newblock ISBN 978-1-4673-7828-4.
\newblock \doi{10.1109/PST.2015.7232973}.
\newblock URL \url{https://doi.org/10.1109/PST.2015.7232973}.

\bibitem[Borgstr{\"{o}}m et~al.(2016)Borgstr{\"{o}}m, Lago, Gordon, and
  Szymczak]{DBLP:conf/icfp/BorgstromLGS16}
Johannes Borgstr{\"{o}}m, Ugo~Dal Lago, Andrew~D. Gordon, and Marcin Szymczak.
\newblock A lambda-calculus foundation for universal probabilistic programming.
\newblock In Jacques Garrigue, Gabriele Keller, and Eijiro Sumii, editors,
  \emph{Proceedings of the 21st {ACM} {SIGPLAN} International Conference on
  Functional Programming, {ICFP} 2016, Nara, Japan, September 18-22, 2016},
  pages 33--46. {ACM}, 2016.
\newblock ISBN 978-1-4503-4219-3.
\newblock \doi{10.1145/2951913.2951942}.
\newblock URL \url{https://doi.org/10.1145/2951913.2951942}.

\bibitem[Bori{\v{c}}i{\'c}(2016)]{bor16}
Marija Bori{\v{c}}i{\'c}.
\newblock Inference rules for probability logic.
\newblock \emph{Publications de l'Institut Math{\'e}matique}, 100\penalty0
  (114):\penalty0 77--86, 2016.

\bibitem[Bori{\v{c}}i{\'c}(2017)]{bor17}
Marija Bori{\v{c}}i{\'c}.
\newblock Suppes-style sequent calculus for probability logic.
\newblock \emph{Journal of Logic and Computation}, 27\penalty0 (4):\penalty0
  1157--1168, 2017.

\bibitem[Bori{\v{c}}i{\'c}(2019)]{borivcic2019sequent}
Marija Bori{\v{c}}i{\'c}.
\newblock Sequent calculus for classical logic probabilized.
\newblock \emph{Archive for Mathematical Logic}, 58\penalty0 (1-2):\penalty0
  119--136, 2019.

\bibitem[Ceolin and Primiero(2019)]{DBLP:conf/ifiptm/CeolinP19}
Davide Ceolin and Giuseppe Primiero.
\newblock A granular approach to source trustworthiness for negative trust
  assessment.
\newblock In Weizhi Meng, Piotr Cofta, Christian~Damsgaard Jensen, and Tyrone
  Grandison, editors, \emph{Trust Management {XIII} - 13th {IFIP} {WG} 11.11
  International Conference, {IFIPTM} 2019, Copenhagen, Denmark, July 17-19,
  2019, Proceedings}, volume 563 of \emph{{IFIP} Advances in Information and
  Communication Technology}, pages 108--121. Springer, 2019.
\newblock ISBN 978-3-030-33715-5.
\newblock \doi{10.1007/978-3-030-33716-2\_9}.
\newblock URL \url{https://doi.org/10.1007/978-3-030-33716-2\_9}.

\bibitem[Ceolin et~al.(2021)Ceolin, Doneda, and
  Primiero]{DBLP:conf/goodit/CeolinDP21}
Davide Ceolin, Francesca Doneda, and Giuseppe Primiero.
\newblock Computable trustworthiness ranking of medical experts in italy during
  the sars-cov-19 pandemic.
\newblock In Ombretta Gaggi, Pietro Manzoni, and Claudio~E. Palazzi, editors,
  \emph{GoodIT '21: Conference on Information Technology for Social Good, Roma,
  Italy, September 9-11, 2021}, pages 271--276. {ACM}, 2021.
\newblock \doi{10.1145/3462203.3475907}.
\newblock URL \url{https://doi.org/10.1145/3462203.3475907}.

\bibitem[Cook(2018)]{97833199614533}
Byron Cook.
\newblock Formal reasoning about the security of amazon web services.
\newblock In Hana Chockler and Georg Weissenbacher, editors, \emph{Computer
  Aided Verification}, pages 38--47, Cham, 2018. Springer International
  Publishing.
\newblock ISBN 978-3-319-96145-3.

\bibitem[Dahlqvist and Kozen(2020)]{DBLP:journals/pacmpl/DahlqvistK20}
Fredrik Dahlqvist and Dexter Kozen.
\newblock Semantics of higher-order probabilistic programs with conditioning.
\newblock \emph{Proc. {ACM} Program. Lang.}, 4\penalty0 ({POPL}):\penalty0
  57:1--57:29, 2020.
\newblock \doi{10.1145/3371125}.
\newblock URL \url{https://doi.org/10.1145/3371125}.

\bibitem[D'Asaro and Primiero(2021)]{DBLP:conf/atal/DAsaroP21}
Fabio~Aurelio D'Asaro and Giuseppe Primiero.
\newblock Probabilistic typed natural deduction for trustworthy computations.
\newblock In Dongxia Wang, Rino Falcone, and Jie Zhang, editors,
  \emph{Proceedings of the 22nd International Workshop on Trust in Agent
  Societies {(TRUST} 2021) Co-located with the 20th International Conferences
  on Autonomous Agents and Multiagent Systems {(AAMAS} 2021), London, UK, May
  3-7, 2021}, volume 3022 of \emph{{CEUR} Workshop Proceedings}. CEUR-WS.org,
  2021.
\newblock URL \url{http://ceur-ws.org/Vol-3022/paper3.pdf}.

\bibitem[de~Amorim et~al.(2021)de~Amorim, Kozen, Mardare, Panangaden, and
  Roberts]{DBLP:conf/lics/AmorimKMPR21}
Pedro H.~Azevedo de~Amorim, Dexter Kozen, Radu Mardare, Prakash Panangaden, and
  Michael Roberts.
\newblock Universal semantics for the stochastic {\(\lambda\)}-calculus.
\newblock In \emph{36th Annual {ACM/IEEE} Symposium on Logic in Computer
  Science, {LICS} 2021, Rome, Italy, June 29 - July 2, 2021}, pages 1--12.
  {IEEE}, 2021.
\newblock ISBN 978-1-6654-4895-6.
\newblock \doi{10.1109/LICS52264.2021.9470747}.
\newblock URL \url{https://doi.org/10.1109/LICS52264.2021.9470747}.

\bibitem[Demolombe(2004)]{DBLP:conf/itrust/Demolombe04}
Robert Demolombe.
\newblock Reasoning about trust: {A} formal logical framework.
\newblock In Christian~Damsgaard Jensen, Stefan Poslad, and Theodosis
  Dimitrakos, editors, \emph{Trust Management, Second International Conference,
  iTrust 2004, Oxford, UK, March 29 - April 1, 2004, Proceedings}, volume 2995
  of \emph{Lecture Notes in Computer Science}, pages 291--303. Springer, 2004.
\newblock ISBN 3-540-21312-0.
\newblock \doi{10.1007/978-3-540-24747-0\_22}.
\newblock URL \url{https://doi.org/10.1007/978-3-540-24747-0\_22}.

\bibitem[Di~Pierro(2020)]{pierro2020atypetheory}
Alessandra Di~Pierro.
\newblock A type theory for probabilistic--calculus.
\newblock In \emph{From Lambda Calculus to Cybersecurity Through Program
  Analysis: Essays Dedicated to Chris Hankin on the Occasion of His
  Retirement}, pages 86--102. Springer, 2020.

\bibitem[Drawel et~al.(2017)Drawel, Bentahar, and
  Shakshuki]{DBLP:conf/ant/DrawelBS17}
Nagat Drawel, Jamal Bentahar, and Elhadi~M. Shakshuki.
\newblock Reasoning about trust and time in a system of agents.
\newblock In Elhadi~M. Shakshuki, editor, \emph{The 8th International
  Conference on Ambient Systems, Networks and Technologies {(ANT} 2017) / The
  7th International Conference on Sustainable Energy Information Technology
  {(SEIT} 2017), 16-19 May 2017, Madeira, Portugal}, volume 109 of
  \emph{Procedia Computer Science}, pages 632--639. Elsevier, 2017.
\newblock \doi{10.1016/j.procs.2017.05.369}.
\newblock URL \url{https://doi.org/10.1016/j.procs.2017.05.369}.

\bibitem[Drawel et~al.(2020)Drawel, Qu, Bentahar, and
  Shakshuki]{DBLP:journals/fgcs/DrawelQBS20}
Nagat Drawel, Hongyang Qu, Jamal Bentahar, and Elhadi~M. Shakshuki.
\newblock Specification and automatic verification of trust-based multi-agent
  systems.
\newblock \emph{Future Gener. Comput. Syst.}, 107:\penalty0 1047--1060, 2020.
\newblock \doi{10.1016/j.future.2018.01.040}.
\newblock URL \url{https://doi.org/10.1016/j.future.2018.01.040}.

\bibitem[Dreossi et~al.(2019)Dreossi, Fremont, Ghosh, Kim, Ravanbakhsh,
  Vazquez-Chanlatte, and Seshia]{dreossi2019verifai}
Tommaso Dreossi, Daniel~J. Fremont, Shromona Ghosh, Edward Kim, Hadi
  Ravanbakhsh, Marcell Vazquez-Chanlatte, and Sanjit~A. Seshia.
\newblock Verifai: A toolkit for the design and analysis of artificial
  intelligence-based systems, 2019.

\bibitem[Gao et~al.(2019)Gao, Hajinezhad, Zhang, Kantaros, and
  Zavlanos]{3302509.3311053}
Qitong Gao, Davood Hajinezhad, Yan Zhang, Yiannis Kantaros, and Michael~M.
  Zavlanos.
\newblock Reduced variance deep reinforcement learning with temporal logic
  specifications.
\newblock In \emph{Proceedings of the 10th ACM/IEEE International Conference on
  Cyber-Physical Systems}, ICCPS '19, pages 237--248, New York, NY, USA, 2019.
  Association for Computing Machinery.
\newblock ISBN 9781450362856.
\newblock \doi{10.1145/3302509.3311053}.
\newblock URL \url{https://doi.org/10.1145/3302509.3311053}.

\bibitem[Genco and Primiero(2023)]{230200958}
Francesco~A. Genco and Giuseppe Primiero.
\newblock A typed lambda-calculus for establishing trust in probabilistic
  programs.
\newblock \emph{CoRR}, abs/2302.00958, 2023.
\newblock \doi{10.48550/arXiv.2302.00958}.
\newblock URL \url{https://doi.org/10.48550/arXiv.2302.00958}.

\bibitem[Ghilezan et~al.(2018)Ghilezan, Iveti{\'c}, Ka{\v{s}}terovi{\'c},
  Ognjanovi{\'c}, and Savi{\'c}]{ghilezan2018probabilistic}
Silvia Ghilezan, Jelena Iveti{\'c}, Simona Ka{\v{s}}terovi{\'c}, Zoran
  Ognjanovi{\'c}, and Nenad Savi{\'c}.
\newblock Probabilistic reasoning about simply typed lambda terms.
\newblock In \emph{International Symposium on Logical Foundations of Computer
  Science}, pages 170--189. Springer, 2018.

\bibitem[Harrison(2003)]{1210044}
J.~Harrison.
\newblock Formal verification at intel.
\newblock In \emph{18th Annual IEEE Symposium of Logic in Computer Science,
  2003}, pages 45--54, 2003.
\newblock \doi{10.1109/LICS.2003.1210044}.

\bibitem[Herzig et~al.(2010)Herzig, Lorini, H{\"{u}}bner, and
  Vercouter]{DBLP:journals/igpl/HerzigLHV10}
Andreas Herzig, Emiliano Lorini, Jomi~Fred H{\"{u}}bner, and Laurent Vercouter.
\newblock A logic of trust and reputation.
\newblock \emph{Log. J. {IGPL}}, 18\penalty0 (1):\penalty0 214--244, 2010.
\newblock \doi{10.1093/jigpal/jzp077}.
\newblock URL \url{https://doi.org/10.1093/jigpal/jzp077}.

\bibitem[Kubyshkina and Primiero(2024)]{pk24}
Ekaterina Kubyshkina and Giuseppe Primiero.
\newblock A possible worlds semantics for trustworthy non-deterministic
  computations.
\newblock \emph{International Journal of Approximate Reasoning}, 172:\penalty0
  109212, 2024.
\newblock ISSN 0888-613X.
\newblock \doi{https://doi.org/10.1016/j.ijar.2024.109212}.
\newblock URL
  \url{https://www.sciencedirect.com/science/article/pii/S0888613X24000999}.

\bibitem[Kwiatkowska et~al.(2002)Kwiatkowska, Norman, and Parker]{647810738106}
Marta~Z. Kwiatkowska, Gethin Norman, and David Parker.
\newblock Prism: Probabilistic symbolic model checker.
\newblock In \emph{Proceedings of the 12th International Conference on Computer
  Performance Evaluation, Modelling Techniques and Tools}, TOOLS '02, pages
  200--204, Berlin, Heidelberg, 2002. Springer-Verlag.
\newblock ISBN 3540435395.

\bibitem[Liau(2003)]{DBLP:journals/ai/Liau03}
Churn{-}Jung Liau.
\newblock Belief, information acquisition, and trust in multi-agent systems--a
  modal logic formulation.
\newblock \emph{Artif. Intell.}, 149\penalty0 (1):\penalty0 31--60, 2003.
\newblock \doi{10.1016/S0004-3702(03)00063-8}.
\newblock URL \url{https://doi.org/10.1016/S0004-3702(03)00063-8}.

\bibitem[Liu and Lorini(2017)]{DBLP:conf/prima/LiuL17}
Fenrong Liu and Emiliano Lorini.
\newblock Reasoning about belief, evidence and trust in a multi-agent setting.
\newblock In Bo~An, Ana L.~C. Bazzan, Jo{\~{a}}o Leite, Serena Villata, and
  Leendert W.~N. van~der Torre, editors, \emph{{PRIMA} 2017: Principles and
  Practice of Multi-Agent Systems - 20th International Conference, Nice,
  France, October 30 - November 3, 2017, Proceedings}, volume 10621 of
  \emph{Lecture Notes in Computer Science}, pages 71--89. Springer, 2017.
\newblock ISBN 978-3-319-69130-5.
\newblock \doi{10.1007/978-3-319-69131-2\_5}.
\newblock URL \url{https://doi.org/10.1007/978-3-319-69131-2\_5}.

\bibitem[Primiero(2020)]{DBLP:journals/jancl/Primiero20}
Giuseppe Primiero.
\newblock A logic of negative trust.
\newblock \emph{J. Appl. Non Class. Logics}, 30\penalty0 (3):\penalty0
  193--222, 2020.
\newblock \doi{10.1080/11663081.2020.1789404}.
\newblock URL \url{https://doi.org/10.1080/11663081.2020.1789404}.

\bibitem[Primiero and Boender(2017)]{DBLP:conf/ifiptm/PrimieroB17}
Giuseppe Primiero and Jaap Boender.
\newblock Managing software uninstall with negative trust.
\newblock In Jan{-}Philipp Stegh{\"{o}}fer and Babak Esfandiari, editors,
  \emph{Trust Management {XI} - 11th {IFIP} {WG} 11.11 International
  Conference, {IFIPTM} 2017, Gothenburg, Sweden, June 12-16, 2017,
  Proceedings}, volume 505 of \emph{{IFIP} Advances in Information and
  Communication Technology}, pages 79--93. Springer, 2017.
\newblock ISBN 978-3-319-59170-4.
\newblock \doi{10.1007/978-3-319-59171-1\_7}.
\newblock URL \url{https://doi.org/10.1007/978-3-319-59171-1\_7}.

\bibitem[Primiero and Boender(2018)]{DBLP:journals/wias/PrimieroB18}
Giuseppe Primiero and Jaap Boender.
\newblock Negative trust for conflict resolution in software management.
\newblock \emph{Web Intell.}, 16\penalty0 (4):\penalty0 251--271, 2018.
\newblock \doi{10.3233/WEB-180393}.
\newblock URL \url{https://doi.org/10.3233/WEB-180393}.

\bibitem[Primiero and D'Asaro(2022)]{DBLP:conf/aiia/PrimieroD22}
Giuseppe Primiero and Fabio~Aurelio D'Asaro.
\newblock Proof-checking bias in labeling methods.
\newblock In Guido Boella, Fabio~Aurelio D'Asaro, Abeer Dyoub, and Giuseppe
  Primiero, editors, \emph{Proceedings of 1st Workshop on Bias, Ethical AI,
  Explainability and the Role of Logic and Logic Programming {(BEWARE} 2022)
  co-located with the 21th International Conference of the Italian Association
  for Artificial Intelligence (AI*IA 2022), Udine, Italy, December 2, 2022},
  volume 3319 of \emph{{CEUR} Workshop Proceedings}, pages 9--19. CEUR-WS.org,
  2022.
\newblock URL \url{https://ceur-ws.org/Vol-3319/paper1.pdf}.

\bibitem[Primiero et~al.(2017)Primiero, Raimondi, Chen, and
  Nagarajan]{DBLP:conf/eurosp/PrimieroRCN17}
Giuseppe Primiero, Franco Raimondi, Taolue Chen, and Rajagopal Nagarajan.
\newblock A proof-theoretic trust and reputation model for {VANET}.
\newblock In \emph{2017 {IEEE} European Symposium on Security and Privacy
  Workshops, EuroS{\&}P Workshops 2017, Paris, France, April 26-28, 2017},
  pages 146--152. {IEEE}, 2017.
\newblock ISBN 978-1-5386-2244-5.
\newblock \doi{10.1109/EuroSPW.2017.64}.
\newblock URL \url{https://doi.org/10.1109/EuroSPW.2017.64}.

\bibitem[Seshia et~al.(2018)Seshia, Desai, Dreossi, Fremont, Ghosh, Kim,
  Shivakumar, Vazquez-Chanlatte, and Yue]{97830300109042}
Sanjit~A. Seshia, Ankush Desai, Tommaso Dreossi, Daniel~J. Fremont, Shromona
  Ghosh, Edward Kim, Sumukh Shivakumar, Marcell Vazquez-Chanlatte, and Xiangyu
  Yue.
\newblock Formal specification for deep neural networks.
\newblock In Shuvendu~K. Lahiri and Chao Wang, editors, \emph{Automated
  Technology for Verification and Analysis}, pages 20--34, Cham, 2018. Springer
  International Publishing.

\bibitem[Singh(2011)]{DBLP:conf/atal/Singh11a}
Munindar~P. Singh.
\newblock Trust as dependence: a logical approach.
\newblock In Liz Sonenberg, Peter Stone, Kagan Tumer, and Pinar Yolum, editors,
  \emph{10th International Conference on Autonomous Agents and Multiagent
  Systems {(AAMAS} 2011), Taipei, Taiwan, May 2-6, 2011, Volume 1-3}, pages
  863--870. {IFAAMAS}, 2011.
\newblock ISBN 978-0-9826571-5-7.
\newblock URL
  \url{http://portal.acm.org/citation.cfm?id=2031741\&CFID=54178199\&CFTOKEN=61392764}.

\bibitem[Termine et~al.(2021{\natexlab{a}})Termine, Antonucci, Primiero, and
  Facchini]{DBLP:conf/eumas/TermineAPF21}
Alberto Termine, Alessandro Antonucci, Giuseppe Primiero, and Alessandro
  Facchini.
\newblock Logic and model checking by imprecise probabilistic interpreted
  systems.
\newblock In Ariel Rosenfeld and Nimrod Talmon, editors, \emph{Multi-Agent
  Systems - 18th European Conference, {EUMAS} 2021, Virtual Event, June 28-29,
  2021, Revised Selected Papers}, volume 12802 of \emph{Lecture Notes in
  Computer Science}, pages 211--227. Springer, 2021{\natexlab{a}}.
\newblock ISBN 978-3-030-82253-8.
\newblock \doi{10.1007/978-3-030-82254-5\_13}.
\newblock URL \url{https://doi.org/10.1007/978-3-030-82254-5\_13}.

\bibitem[Termine et~al.(2021{\natexlab{b}})Termine, Primiero, and
  D'Asaro]{DBLP:conf/lori/TerminePD21}
Alberto Termine, Giuseppe Primiero, and Fabio~Aurelio D'Asaro.
\newblock Modelling accuracy and trustworthiness of explaining agents.
\newblock In Sujata Ghosh and Thomas Icard, editors, \emph{Logic, Rationality,
  and Interaction - 8th International Workshop, {LORI} 2021, Xi'ian, China,
  October 16-18, 2021, Proceedings}, volume 13039 of \emph{Lecture Notes in
  Computer Science}, pages 232--245. Springer, 2021{\natexlab{b}}.
\newblock ISBN 978-3-030-88707-0.
\newblock \doi{10.1007/978-3-030-88708-7\_19}.
\newblock URL \url{https://doi.org/10.1007/978-3-030-88708-7\_19}.

\bibitem[Urban and Min\`{e}(2021)]{urban2021review}
Caterina Urban and Antoine Min\`{e}.
\newblock A review of formal methods applied to machine learning, 2021.

\bibitem[Warrell(2016)]{DBLP:journals/corr/Warrell16}
Jonathan~H. Warrell.
\newblock A probabilistic dependent type system based on non-deterministic beta
  reduction.
\newblock \emph{CoRR}, abs/1602.06420, 2016.
\newblock URL \url{http://arxiv.org/abs/1602.06420}.

\bibitem[Wing(2021)]{3448248}
Jeannette~M. Wing.
\newblock Trustworthy ai.
\newblock \emph{Commun. ACM}, 64\penalty0 (10):\penalty0 64--71, sep 2021.
\newblock ISSN 0001-0782.
\newblock \doi{10.1145/3448248}.
\newblock URL \url{https://doi.org/10.1145/3448248}.

\end{thebibliography}

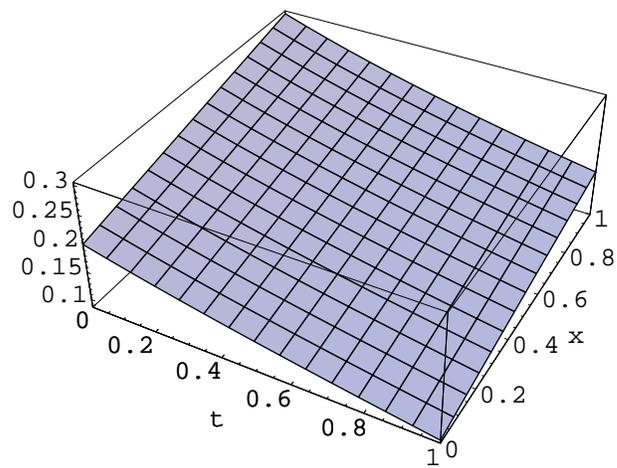
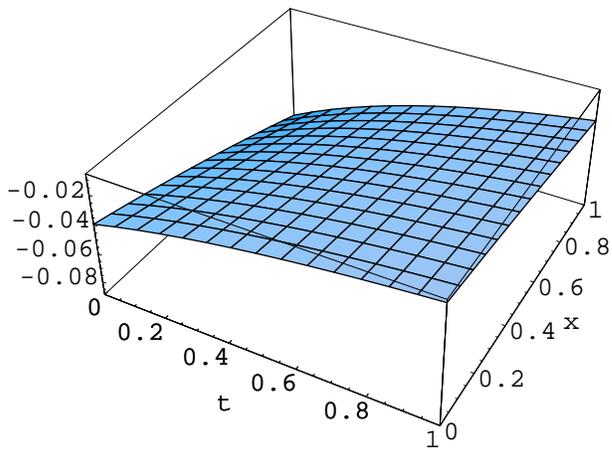
\end{document}